\documentclass[journal]{IEEEtran}
\usepackage{setspace}
\usepackage{amsthm}
\usepackage{epsfig}
\usepackage{amsmath}
\usepackage{amsfonts}
\usepackage{color}
\usepackage{lettrine}
\usepackage{mathtools}
\usepackage{graphicx}
\usepackage{float}
\usepackage{array}
\usepackage{textcomp}
\usepackage{multicol}
\usepackage[space]{grffile}
\usepackage{url}

\newtheorem{lemma}{Lemma}
\newtheorem{proposition}{Proposition}

\newcommand{\E}{\ensuremath{\mathbb E}}

\newcommand{\jj}{$ \mathcal{J} $}
\newcommand{\rr}{$ \mathcal{R} $}
\newcommand{\src}{$ \mathcal{S}_1 $}
\newcommand{\des}{$ \mathcal{S}_2 $}

\def \treq {\stackrel{\tiny \Delta}{=}}

\begin{document}

\title{\hspace{-2mm}Secure Two-Way Transmission via Wireless-Powered \\ Untrusted Relay and External Jammer}
%
%
%
%
	
\author{Milad Tatar Mamaghani,~Ali Kuhestani,~\IEEEmembership{Student~Member,~IEEE},~and Kai-Kit Wong,~\IEEEmembership{Fellow,~IEEE}
		
\thanks{This work is supported in part under EPSRC grant EP/K015893/1.}

\thanks{M. T. Mamaghani and A. Kuhestani are with the Department of Electrical
	Engineering, Amirkabir University of Technology,~Tehran,~Iran.~(e-mail:~\{m.tatarmamaghani,~a.kuhestani\}@aut.ac.ir).}
\thanks{K.-K. Wong is with the Department of Electronic and Electrical Engineering, University College
	London, London WC1E 6BT, U.K.~(e-mail: kaikit.
	wong@ucl.ac.uk).}
}

\maketitle
\markboth{}{}

\begin{abstract}
In this paper, we propose a two-way secure communication scheme where two transceivers exchange confidential messages via a wireless powered untrusted amplify-and-forward (AF) relay in the presence of an external jammer. We take into account both friendly jamming (FJ) and Gaussian noise jamming (GNJ) scenarios. Based on the time switching (TS) architecture at the relay, the data transmission is done in three phases. In the first phase, both the energy-starved nodes, the untrustworthy relay and the jammer, are charged by non-information radio frequency (RF) signals from the sources. In the second phase, the two sources send their information signals and concurrently, the jammer transmits artificial noise to confuse the curious relay. Finally, the third phase is dedicated to forward a scaled version of the received signal from the relay to the sources. For the proposed secure transmission schemes, we derive new closed-form lower-bound expressions for the ergodic secrecy sum rate (ESSR) in the high signal-to-noise ratio (SNR) regime. We further analyze the asymptotic ESSR to determine the key parameters; the high SNR slope and the high SNR power offset of the jamming based scenarios. To highlight the performance advantage of the proposed FJ, we also examine the scenario of without jamming (WoJ).  Finally, numerical examples and discussions are provided to acquire some engineering insights, and to demonstrate the impacts of different system parameters on the secrecy performance of the considered communication scenarios. The numerical results illustrate that the proposed FJ significantly outperforms the traditional one-way communication and the Constellation rotation approach, as well as our proposed benchmarks, the two-way WoJ and GNJ scenarios.\\
\end{abstract}
\begin{IEEEkeywords}
	Wireless power transfer, Physical layer security, Two-way communication, Untrusted relaying, Jammer
\end{IEEEkeywords}

\section{Introduction}
\subsection{Background and Motivation}
\lettrine[lines=2]{C}{ ooperative} relaying improves energy efficiency, extends coverage, and increases the throughput of wireless communication networks. Accordingly, in recent years, the benefits of relaying have been viewed from the standpoint of wireless physical-layer security (PLS) {\cite{review0}} which has been recognized as an emerging design paradigm to enhance the security of next generation wireless networks {\cite {review1}}. 
In the context of relaying-based transmission networks, a key area of interest is the untrusted relaying where the source-destination communication is assisted by a relay which may also be a potential eavesdropper {\cite {review1}}, {\cite {review2}}. In practice, untrusted relaying scenario may occur in large-scale wireless systems such as heterogeneous networks, device-to-device (D2D) communications and Internet-of-things (IoT) applications, where the data of sources are often retransmitted by several intermediate nodes with low security clearance.

Secure transmission employing an untrusted relay was first studied in {\cite {Oohama}}, where an achievable non-zero secrecy rate is obtained through jamming signal transmission. To be specific, two general types of jamming signals have been proposed in the literature to improve the PLS of wireless networks: 1)  friendly jamming (FJ) and 2) Gaussian noise jamming (GNJ). In the former, the jamming signal is a priori known at the legal receiver \cite{review0}--\cite{Fang2013}, while in the latter, the legitimate receiver has no information about the jamming signal and hence, the receiver considers the jamming signal as an interfering signal \cite{Hanzo2016}, \cite{KWang2017}. We mention that FJ offers better secrecy performance compared to GNJ, due to the fact that the legitimate receiver cancels the pre-defined jamming signal. Of course, this performance advantage is obtained at the cost of higher  implementational complexity to the
network. In the area of untrusted relaying, for the first time, the authors in {\cite {he1}} proposed destination-based cooperative jamming (DBCJ) technique to achieve a positive secrecy rate for a one-way untrusted relay, in which the jammer is co-located with the destination receiver. Motivated by the pioneering work {\cite {he1}},  a great deal of research has been dedicated in the field of one-way untrusted relaying {\cite {Wang2014}}--{\cite {kadd}}. 

Recently, several works have considered the more interesting scenario of two-way untrusted relaying {\cite {zhang0}}--{\cite {Xu2015}}, where physical-layer network coding can enhance the security of communication by receiving a superimposed signal from the two sources instead of each individual signal. The authors in \cite{zhang0} proposed a game-theoretic power control scheme between the two sources and multiple jammers, where
single user decoding (SUD) is assumed for the untrusted relay to extract the information signal. We note that, in the SUD operation, the relay attempts to decode one message while the other signal is considered as an interference. However, the untrusted relay in two-way relaying can potentially eavesdrop the legitimate transmissions according to another advanced strategy namely multi-user decoding (MUD), in which  the relay attempts to decode two information signals transmitted by the two sources. It is worth noting that the MUD can be considered as the worst case scenario in untrusted relaying networks \cite{Tao2}, \cite{Huang3}. In \cite{Tao2}, the authors proposed iterative algorithms to jointly optimize the precoding vector at the multiple antenna sources and the precoding vector at the multiple antenna MUD relaying network such that the instantaneous secrecy sum rate without friendly jammer is
maximized. Then, a joint optimization of transmit covariance matrices and relay selection was proposed in \cite{Huang3} for a two-way MUD relaying network and in the absence of a friendly jammer. The proposed optimal algorithm in \cite{Huang3} is solved through the semi-definite programming combined with a line search method and thus suffers from high computational complexity. Xu \emph{et al.} in \cite{Xu2015} proposed a new secure transmission protocol based on constellation rotation approach in the presence of a SUD untrusted relay and without employing any jammer. Finally, optimal power allocation and secrecy sum rate analysis in the two-way untrusted relaying conducting MUD has been studied in \cite{kuh2017}. The authors in \cite{kuh2017} highlighted that FJ scenario improves the secrecy performance significantly compared to without employing an external jammer.

A paramount issue in many wireless communication applications is this fact that some of communication nodes may not have access to permanent power sources due to mobility. Furthermore, frequent recharging and replacement of batteries would be inconvenient in certain circumstances; e.g., in wireless body area network applications, where  medical devices are required to be implanted inside patients' body. For such network, energy harvesting (EH) from ambient resources, e.g., solar and wind has been introduced as a promising approach to prolong the lifetime of energy-constrained wireless nodes \cite{Ku}. However, conventional EH methods are usually uncontrollable, and thus may not satisfy the quality of service (QoS) requirement of wireless networks. To overcome this issue, a new type of EH solution called wireless information and power transfer (WIPT) was introduced in \cite{Grover}. The key idea behind WIPT is to capture radio frequency (RF) signal propagated by a source node and then converting the RF signal to direct current to charge its battery, and also for signal processing or information transmission. In the area of cooperative networks, two main relaying protocols, i.e., time switching (TS) and power splitting (PS) policies have been proposed to implement the WIPT technology. In recent research, great efforts have been dedicated to the study of WIPT for non-security based {\cite{Nasir}}, {\cite{Al-Hraishawi}} and security based systems {\cite {salman}}, {\cite {Kalamkar}}. To be specific, the authors in \cite{salman} proposed employing a wireless-powered jammer to provide secure communication between a source and a destination. Then, the authors in \cite{salman} derived a closed-form expression for the throughput, and characterized the long-term behavior of the proposed protocol. In untrusted relaying networks, Kalamkar \emph{et al.} in \cite{Kalamkar} studied secure one-way communication in the presence of an  untrusted relay based on WIPT technology, where either TS or PS is adopted at the relay.

\subsection{Our Contributions and Key Results}
In contrast to the aforementioned works, in this paper we take into account the PLS of a two-way amplify-and-forward
(AF) relaying, where two sources exchange  confidential messages using an untrustworthy MUD relay with the help of an external jammer to enhance the PLS. A self-reliant cooperative wireless network is proposed in which the relay and jammer as energy-starved helping devices are powered with wireless energy of RF signals. We assume that the TS receiver architecture is adopted at both the relay and jammer. The role of the relay is to harvest the energy in order to forward the received information signal to the sources, while the mission of the jammer is to utilize the harvested energy to degrade the wiretap channel of the untrusted relay. For this proposed secure transmission scheme, we derive new tight lower-bound expressions for the ergodic secrecy sum rate (ESSR) of the following three scenarios in the high signal-to-noise ratio (SNR) regime: 1) Without jamming
(WoJ), where the jammer is not activated, 2) FJ, where the jamming signal is known a priori at the two sources, and 3) GNJ, where the jamming signal is unknown at the sources. We further characterize the high SNR slope and the high SNR power offset for the ESSR of the three WoJ, FJ, and GNJ scenarios, to explicitly determine the impact of network parameters on the ESSR \cite{Lozano}. Based on our analytical results, we further highlight the impact of several system design parameters including the EH time ratio, power allocation factor, transmit SNR, nodes distance, and path loss exponent on the ESSR performance. Numerical examples  show that the proposed two-way FJ  provides significantly better ESSR compared with its traditional counterparts namely the one-way communication \cite{Kalamkar} and the two-way constellation rotation (CR) based communication \cite{Xu2015}, as well as our proposed  WoJ and GNJ schemes. We also observe that unlike the ESSR performance of WoJ, FJ, one-way communication, and CR, the ESSR of GNJ scenario is limited to a secrecy rate ceiling in the high SNR regime. This interesting observation indicates the importance of sharing a pre-defined jamming signal between the two sources.

Our work is different from the following most related papers:  While the authors in \cite{Nasir}, considered a point-to-point communication based on the WIPT strategies via a relay, they investigated the throughput analysis. Unlike \cite{Nasir}, we adopt the WIPT technology to develop an EH based communication network under the constraint of secure transmission. Therefore, the use of EH in \cite{Nasir} is fundamentally different from our work.  Different from \cite{zhang0}, \cite{kuh2017}, in this work, we  consider using wireless powered
relay and jammer to help the secure communication. Different from \cite{zhang0}, we  assume that the MUD is adopted at the untrusted relay to consider the worst-case scenario in our network. It is worth noting that this is the first paper studying the GNJ scenario in untrusted relaying network. 
In \cite{Kalamkar}, the authors studied the one-way secure transmission  based on wireless EH at the untrusted relay. Different from \cite{Kalamkar}, we consider the two-way untrusted MUD relaying in which two sources are able to exchange their information. Furthermore, we propose to employ  an external jammer to boost the secrecy sum rate of the network. This paper is also fundamentally different from \cite{salman} where a wireless-powered jammer is utilized to facilitate the secure communication between a pair of source-destination nodes. Different from \cite{salman}, in our work, extending the coverage area of transmission by using a relay is undeniable in terms of practicality inasmuch as we assume lack of direct link between the two communication nodes. In other words, in our proposed scheme a relay must be exploited to provide communication. This scenario is applicable for communication networks when two sources are located far apart or within heavily shadowed areas. Therefore, the network design and the performance  analysis of our work is different from \cite{salman}.

The remainder of this paper is organized as follows. System model and the proposed relaying protocols are presented in Section II, followed by signals and powers representation in Section III. Section IV and V investigate the secrecy performance of the proposed protocol and derive  new closed-form expressions for the ESSR, as well as analyze the asymptotic ESSR including the high SNR slope and the high SNR power offset. Simulation results and discussions are detailed in Section VI. Finally, conclusions are given in Section VII.

\begin{figure}
	\centering
	\includegraphics[width= \columnwidth]{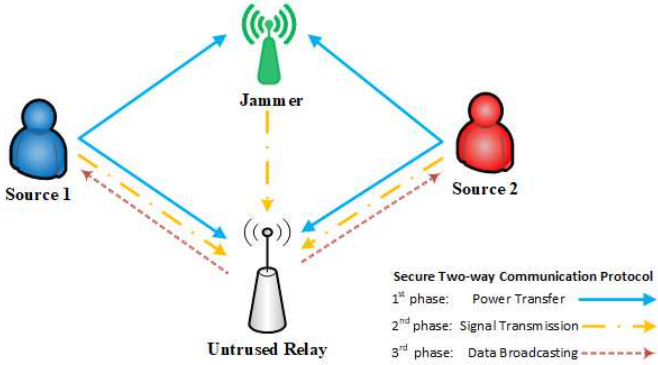}
	\caption{ \small System model of a wireless powered secure two-way network using an untrusted relay and an external jammer.}
	\label{fig}
\end{figure}

\section{System Model}

We consider a two-way communication scenario illustrated in Fig. \ref{fig},  where the two transceivers called ($\mathcal{S}_1$) and ($\mathcal{S}_2$) communicate with each other via an untrusted AF relay ($\mathcal{R}$). In the proposed system, we assume that all the nodes are equipped with a single antenna and operate in half-duplex mode, i.e., sending and receiving data concurrently is not possible. The direct link between~\src~and~\des~is assumed to be unavailable. As such, using the relay service is mandatory {\cite {Nasir}}. Unlike~$\mathcal{S}_1$~and~$\mathcal{S}_2$~that need to decode one signal, we assume that~$\mathcal{R}$ adopts MUD to extract both of the sources' signals. Additionally, the channels between the nodes are assumed to be reciprocal, following a quasi-static block-fading Rayleigh model, where the channel properties remain constant over the block time of one message exchange. We denote $h_{ij}$ as the channel coefficient between the nodes $i$ and $j$, with channel reciprocity where $h_{ij}=h_{ji}$. The channel power gain $|h_{ij}|^2$ follows an exponential distribution with mean $\mu_{ij}$. We also denote $f_{|h_{ij}|^2}(x)$ as the probability density function (PDF) of random variable (RV) $|h_{ij}|^2$. Furthermore, we assume that the sources have perfect knowledge of the channel state information (CSI) of the links~\src\textendash\rr,~\des\textendash\rr, and~\jj\textendash\rr~{\cite {zhang0}}.

Three secure transmission scenarios taken into account in this paper are detailed as follows:

$\bullet$~\emph{WoJ scenario}: To see how employing a jammer can impact on the secrecy performance of the proposed communication network, the WoJ scenario is studied, in which the data transmission policy is as follows. At the beginning,~\rr~is charged by the two sources in the first phase to facilitate the relaying. Next,~\src~and~\des~start to send their superimposed signals to~\rr~in the second phase, followed by forwarding the received data to the sources after amplification by~\rr~during the last phase. Finally, each source decodes the signal of the opposite node. It is worth mentioning that the WoJ brings high simplicity with very low cost compared with the FJ and GNJ scenarios.

$\bullet$~\emph{FJ scenario}: In this scenario, one external jammer ($\mathcal{J}$) is employed to enhance the security of the network by degrading the relay channel capacity through sending its jamming signal. In the FJ scenario, the data exchange between two sources is implemented in three phases. In the first phase, as shown with solid lines in Fig. \ref{fig},~\src~and~\des~transmit non-information signals toward~\jj~and~\rr, to charge them via the RF signals.  Note that both $\mathcal{R}$ and $\mathcal{J}$ are assumed to be energy-starved nodes, yet equipped with rechargeable batteries with infinite capacity. It is also assumed that most of the nodes' energy are consumed for data transmission, and energy consumption for signal processing is ignored for simplicity \cite{Nasir}. During the second phase, the source nodes send their information signals to~\rr. Simultaneously,~\jj~deteriorates the channel capacity of~\rr~by transmitting the jamming signal powered by the sources in the first stage, as demonstrated with dashed lines. Finally, in the third phase,~\rr~broadcasts the scaled version of the received signal to~\src~and~\des,~and then each source extracts its corresponding information signal after self-interference and jamming signal cancellation. In this scenario we assume the sources have perfect knowledge of the jamming signal transmitted by~\jj~for they have paid for the jamming service. This is a common assumption in the FJ literature, e.g., \cite{zhang0}, \cite{salman}, \cite{Amarasuriya}, and \cite{srt1}, where the pre-defined jamming signal can be generated by using some pseudo-random codes or some cryptographic signals that are known to both the jammer and the sources but not available to the curious relay. For the sake of this purpose, one method is implementing DBCJ by which transmission of the jamming signal from~\jj~to~\src~and~\des~can be accomplished, leading to the acquisition of the jamming signal solely by the legitimate receivers. Accordingly, before data transmission (and vividly after EH phase),~\jj~sends the specific jamming signal to~\src, while simultaneously,~\src~transmits an artificial noise signal to confuse~\rr. In the next step, the relay broadcasts the degraded received signal and consequently,~\src~can extract the jamming signal. The same procedure can be implemented by~\des~to obtain the jamming signal as well \footnote{Another approach to implement FJ is to use some cryptographic signals at the jammer to jam, where the decryption book is a secret key only accessible to the sources. Then, the sources can have perfect knowledge of the jamming signals if the jammer sends some additional bits consisting of the information of the jamming signals, e.g., pseudo-noise (PN) codes, transmitted to the sources periodically \cite{zhang0}.  To make the PN codes available at both~\jj~and the sources, some exchange among the nodes is undeniable. Therefore, to guarantee that these PN codes cannot be intercepted by the eavesdropper, the aforementioned DBCJ technique can also be utilized. Therefore, combining both the cryptographic techniques and the PLS seems promising approach for secure communication, and we will examine such idea in our future works.}. In this DBCJ based method, we assumed that both the helping devices are charged successfully during the EH phase such that the jamming transmission to the sources is enabled, otherwise the two sources can not obtain a perfect knowledge of the jamming signal, and hence results in the GNJ scenario. 

$\bullet$~\emph{GNJ scenario}: In this scenario, the data transmission protocol  is the same as the FJ. Different from the FJ, the two sources have no knowledge about the jamming signal and therefore, the jamming signal is considered as an interfering signal at~\src~and~\des. Based on this fact, the proposed GNJ network experiences performance loss compared with the FJ scenario. In contrast to FJ that the secrecy performance advantage is obtained at the cost of higher implementational complexity, the GNJ scenario enjoys having little workload of online computation.

\subsection{Time Switching Relaying Protocol}

\begin{figure}[t]
	\centering
	\includegraphics[width= \columnwidth]{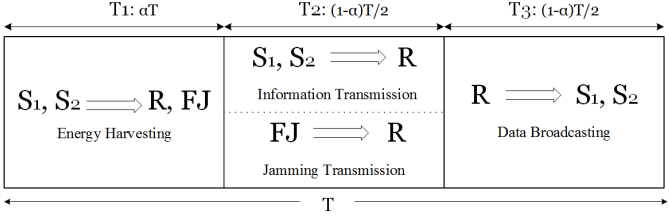}
	\caption{ \small Time switching relaying protocol for two-way secure communication via a wireless powered untrusted relay and a jammer.}
	\label{fig2}
\end{figure}

Fig. \ref{fig2} describes the proposed wireless EH two-way relaying transmission protocol. Using the TS policy, the relay switches from EH to information encoding, and completes a round of data exchange in three phases over a period of $T$. To be specific, in the first phase with the duration of $T_1=\alpha T$ ($ 0<\alpha<1$), both~\rr~and~\jj~harvest the energy of the RF signals transmitted by~\src~and~\des. In the second time slot which lasts $T_2=(1-\alpha)\frac{T}{2}$,~\src~and~\des~send their information signals to~\rr, and simultaneously~\jj~transmits its jamming signal. Finally, in the third phase,~\rr~broadcasts the scaled version of the received signal. It is worth noting that the parameter $\alpha$ which indicates the ratio of EH time to the total transmission time of one period has significant impact on the system performance, i.e., related to the value of $\alpha$, the secrecy rate of the proposed network is changed as will be shown numerically in Section VII.

\section{Signals and Powers Representation}

In the following, the signals and powers corresponding to the WoJ, FJ, and GNJ scenarios are presented. We first denote $x_{\mathcal{S}_{\textit{i}}}$, $\textit{i}\in\{1, 2\}$, and $x_{J}$ as the information signals and the jamming signal with the powers of $P_{\mathcal{S}_{{\textit{i}}}}$ and $P_{TJ}$, respectively.

\subsection{Energy Harvesting at the Relay and Jammer}
\subsubsection{Without Jamming}
In the first phase, the two source nodes send non-information signals, to charge the relay. The received power at~\rr~is given by
\begin{equation}\label{pr}
P_{\mathcal{R}}={P_{\mathcal{S}_{1}}}  |h_{\mathcal{S}_1\mathcal{R}}|^2+{P_{\mathcal{S}_{2}}}  |h_{\mathcal{S}_2\mathcal{R}}|^2.
\end{equation}
Based on the proposed TS protocol, the harvested energy $E_{HR}$ in the duration of $\alpha T$ at~\rr~is given by
\begin{equation}\label{ehr}
E_{HR}=\eta \alpha T({P_{\mathcal{S}_{1}}}  |h_{\mathcal{S}_1\mathcal{R}}|^2+{P_{\mathcal{S}_{2}}}  |h_{\mathcal{S}_2\mathcal{R}}|^2),
\end{equation}
where $\eta$ represents the energy conversion efficiency factor, and $ 0<\eta<1$. The relay uses the harvested energy obtained in the first phase \eqref{ehr} to retransmit the received signal in the third phase with the power $P_{TR}$ which can be written as
\begin{align}\label{ptr}
P_{TR}=\frac{E_{HR}}{(1-\alpha)\frac{T}{2}}=\beta^{-1} ({P_{\mathcal{S}_{1}}}  |h_{\mathcal{S}_1\mathcal{R}}|^2+{P_{\mathcal{S}_{2}}}  |h_{\mathcal{S}_2\mathcal{R}}|^2),
\end{align}
where $\beta$ is defined as $\beta\treq\frac{1-\alpha}{2\eta\alpha}$.

\subsubsection{Friendly Jamming/Gaussian noise jamming}
In the FJ and GNJ scenarios, EH at~\rr~is the same as the WoJ scheme. Similarly, for the received power at~\jj~in the first phase, can be written as
\begin{equation}\label{pj}
P_{\mathcal{J}}={P_{\mathcal{S}_{1}}}  |h_{\mathcal{S}_1\mathcal{J}}|^2+{P_{\mathcal{S}_{2}}}  |h_{\mathcal{S}_2\mathcal{J}}|^2,
\end{equation}
and the amount of harvested energy at~\jj~during one frame of communication can be represented as
\begin{equation}\label{ehj}
E_{HJ}=\eta\alpha T({P_{\mathcal{S}_{1}}}  |h_{\mathcal{S}_1\mathcal{J}}|^2+{P_{\mathcal{S}_{2}}}  |h_{\mathcal{S}_2\mathcal{J}}|^2).
\end{equation}
Furthermore, during the second phase,~\jj~uses the harvested energy in (\ref{ehj}) to transmit its jamming signal with the power of $P_{TJ}$, which can be expressed as 
\begin{align}\label{ptj}
P_{TJ}=\frac{E_{HJ}}{(1-\alpha)\frac{T}{2}}=\beta^{-1} ({P_{\mathcal{S}_{1}}}  |h_{\mathcal{S}_1\mathcal{J}}|^2+{P_{\mathcal{S}_{2}}}  |h_{\mathcal{S}_2\mathcal{J}}|^2).
\end{align}
Note that in the aforementioned scenarios, $P_{\mathcal{R}}$ and $P_{\mathcal{J}}$ should be more than the minimum predefined threshold power ($\Theta$) to activate the harvesting circuitry, unless the helper nodes will remain inactive.

\subsection{Signals Representation}	
\subsubsection{Without Jamming}
For the WoJ scenario, the received signal at~\rr~in the second phase, can be expressed as
\begin{align}\label{yr_woj}
y_{\mathcal{R}}=\sqrt{{P_{\mathcal{S}_{1}}}}x_{\mathcal{S}_1}h_{\mathcal{S}_1\mathcal{R}}+\sqrt{{P_{\mathcal{S}_{2}}}}x_{\mathcal{S}_2}h_{\mathcal{S}_2\mathcal{R}}+n_{\mathcal{R}},
\end{align}
where $n_{\mathcal{R}}$ is considered as a zero-mean additive white Gaussian noise (AWGN) at~\rr. Based on the received signal $y_{\mathcal{R}}$ in (\ref{yr_woj}) and considering the MUD at~\rr,~the SNR at $\mathcal{R}$ can be obtained as
\begin{align} \label{gammaR_woj}
\gamma_{\mathcal{R}}&=\frac{{P_{\mathcal{S}_{1}}}  |h_{\mathcal{S}_1\mathcal{R}}|^2+{P_{\mathcal{S}_{2}}} |h_{\mathcal{S}_2\mathcal{R}}|^2}{ N_{0}},
\end{align}
where $N_{0}$ denotes the power of AWGN at $\mathcal{R}$, and for simplicity the processing noise is ignored {\cite {Kalamkar}}.

Finally, in the third phase, $\mathcal{R}$ broadcasts the amplified version of the received signal which is given by
\begin{equation}\label{output_relay_woj}
x_{\mathcal{R}}=G y_{\mathcal{R}},
\end{equation}
where $G$ is the scaling factor of~\rr~as
\begin{equation}\label{zeta_woj}
G=\sqrt{\frac{P_{TR}}{{P_{\mathcal{S}_{1}}}  |h_{\mathcal{S}_1\mathcal{R}}|^2+{P_{\mathcal{S}_{2}}}  |h_{\mathcal{S}_2\mathcal{R}}|^2+N_{0}}}.
\end{equation}
Next, we focus on the received signal at~\des, from which similar expressions can be derived for the received signal at~\src. By using (\ref{yr_woj}) and (\ref{output_relay_woj}), the received signal at~\des~after self-interference cancellation can be expressed as
\begin{align}\label{ys2_woj}
y_{\mathcal{S}_2}=\sqrt{{P_{\mathcal{S}_{1}}}} G h_{\mathcal{S}_1\mathcal{R}}h_{\mathcal{RS}_{2}}x_{\mathcal{S}_1}+G h_{\mathcal{RS}_{2}}n_{\mathcal{R}}+n_{\mathcal{S}_2}.
\end{align}
Substituting (\ref{zeta_woj}) into (\ref{ys2_woj}), the received instantaneous end-to-end SNR at~\des~after some algebraic manipulations can be obtained as
\begin{align} \label{gammad1_woj}
\gamma_{\mathcal{S}_2}=\frac{{P_{\mathcal{S}_{1}}}  |h_{\mathcal{S}_1\mathcal{R}}|^2 |h_{\mathcal{RS}_{2}}|^2}{N_{0}|h_{\mathcal{RS}_{2}}|^2+N_0\beta+\epsilon},
\end{align}
where $\epsilon=\frac{ N_0^2\beta}{P_{\mathcal{S}_{1}}|h_{\mathcal{S}_1\mathcal{R}}|^2+{P_{\mathcal{S}_{2}}|h_{\mathcal{S}_2\mathcal{R}}|^2}}$. Following the same procedure for calculation of $\gamma_{\mathcal{S}_2}$, the resultant instantaneous end-to-end SNR at $\mathcal{S}_1$ is also given by
\begin{equation}\label{gammad2_woj}
\gamma_{\mathcal{S}_1}=\frac{{P_{\mathcal{S}_{2}}}  |h_{\mathcal{S}_2\mathcal{R}}|^2 |h_{\mathcal{RS}_{1}}|^2}{N_{0}|h_{\mathcal{RS}_{1}}|^2+N_0\beta+\epsilon}.
\end{equation}

\subsubsection{Friendly Jamming}
For the FJ scenario, the received signal at~\rr~in the second phase, can be expressed as
\begin{align}\label{yr_fj}
y_{\mathcal{R}}&=\sqrt{{P_{\mathcal{S}_{1}}}}x_{\mathcal{S}_1}h_{\mathcal{S}_1\mathcal{R}}+\sqrt{{P_{\mathcal{S}_{2}}}}x_{\mathcal{S}_2}h_{\mathcal{S}_2\mathcal{R}} \nonumber \\
&+\sqrt{P_{TJ}}x_{\mathcal{J}}h_{\mathcal{JR}}+n_{\mathcal{R}}.
\end{align}
Substituting $P_{TJ}$ given by \eqref{ptj} into (\ref {yr_fj}),~the SNR at $\mathcal{R}$ can be obtained as
\begin{align} \label{gammaR_fj}
\gamma_{\mathcal{R}}&=\frac{{P_{\mathcal{S}_{1}}}  |h_{\mathcal{S}_1\mathcal{R}}|^2+{P_{\mathcal{S}_{2}}} |h_{\mathcal{S}_2\mathcal{R}}|^2}{P_{TJ} |h_{\mathcal{JR}}|^2+ N_{0}}\nonumber\\
&=\frac{{P_{\mathcal{S}_{1}}}  |h_{\mathcal{S}_1\mathcal{R}}|^2+{P_{\mathcal{S}_{2}}} |h_{\mathcal{S}_2\mathcal{R}}|^2}{\beta^{-1} ({P_{\mathcal{S}_{1}}}  |h_{\mathcal{S}_1\mathcal{J}}|^2+{P_{\mathcal{S}_{2}}}  |h_{\mathcal{S}_2\mathcal{J}}|^2) |h_{\mathcal{JR}}|^2+ N_{0}},
\end{align}
Finally, $\mathcal{R}$ broadcasts the amplified version of the received signal, $x_{\mathcal{R}}=G y_{\mathcal{R}}$, with the amplification factor of
\begin{align}\label{output_relay_fj}
G\hspace{-1mm}=\hspace{-1mm}\sqrt{\frac{P_{TR}}{{P_{\mathcal{S}_{1}}}  |h_{\mathcal{S}_1\mathcal{R}}|^2\hspace{-1mm}+\hspace{-1mm}{P_{\mathcal{S}_{2}}}  |h_{\mathcal{S}_2\mathcal{R}}|^2\hspace{-1mm}+\hspace{-1mm}P_{TJ}|h_{\mathcal{JR}}|^2\hspace{-1mm}+\hspace{-1mm}N_{0}}}.
\end{align}
Moreover, by using (\ref{yr_fj}) and (\ref{output_relay_fj}), the received signal at~\des~can be expressed as
\begin{align}\label{ys2prime_fj}
y_{\mathcal{S}_2}'&=x_{\mathcal{R}}h_{\mathcal{RS}_{2}}+n_{\mathcal{S}_2}\nonumber\\
&=\sqrt{{P_{\mathcal{S}_{1}}}} G h_{\mathcal{S}_1\mathcal{R}}h_{\mathcal{RS}_{2}}x_{\mathcal{S}_1}+\sqrt{{P_{\mathcal{S}_{2}}}} G h_{\mathcal{S}_2\mathcal{R}}h_{\mathcal{RS}_{2}}x_{\mathcal{S}_2} \nonumber\\
&+\sqrt{P_{TJ}}G h_{\mathcal{JR}}h_{\mathcal{RS}_{2}}x_{\mathcal{J}}+G h_{\mathcal{RS}_{2}}n_{\mathcal{R}}+n_{\mathcal{S}_2},
\end{align}
Since the jamming signal in FJ scenario is fully known at the sources, as well as the CSI of the links~\src\textendash\rr,~\des\textendash\rr,~and~\rr\textendash\jj,~ \des~can eliminate the jamming signal and its own self-interference from \eqref{ys2prime_fj}, which simplifies as
\begin{equation}\label{ys2_fj}
y_{\mathcal{S}_2}=\sqrt{{P_{\mathcal{S}_{1}}}} G h_{\mathcal{S}_1\mathcal{R}}h_{\mathcal{RS}_{2}}x_{\mathcal{S}_1}+G h_{\mathcal{RS}_{2}}n_{\mathcal{R}}+n_{\mathcal{S}_2}.
\end{equation}
Substituting (\ref{output_relay_fj}) into (\ref{ys2_fj}), and then using $P_{TJ}$ given by \eqref{ptj}, the received instantaneous end-to-end SNR at~\des~is given by
\begin{align} \label{gammad1_fj}
\gamma_{\mathcal{S}_2}\hspace{-1mm}=\hspace{-1mm}\frac{{P_{\mathcal{S}_{1}}}  |h_{\mathcal{S}_1\mathcal{R}}|^2 |h_{\mathcal{RS}_{2}}|^2}{N_{0}|h_{\mathcal{RS}_{2}}|^2\hspace{-1mm}+\hspace{-1mm}\frac{N_0\big(P_{\mathcal{S}_1}|h_{\mathcal{S}_{1}\mathcal{J}}|^2+P_{S_{2}}|h_{\mathcal{S}_{2}\mathcal{J}}|^2\big)|h_{\mathcal{JR}}|^2}{P_{\mathcal{S}_{1}}|h_{\mathcal{S}_1\mathcal{R}}|^2+{P_{\mathcal{S}_{2}}|h_{\mathcal{S}_2\mathcal{R}}|^2}}\hspace{-1mm}+\hspace{-1mm}N_0\beta\hspace{-1mm}+\hspace{-1mm}\epsilon},
\end{align}
Similarly, the received SNR at $\mathcal{S}_1$ is obtained as
\begin{align}\label{gammad2_fj}
\gamma_{\mathcal{S}_1}\hspace{-1mm}=\hspace{-1mm}\frac{{P_{\mathcal{S}_{2}}}  |h_{\mathcal{S}_2\mathcal{R}}|^2 |h_{\mathcal{RS}_{1}}|^2}{N_{0}|h_{\mathcal{RS}_{1}}|^2\hspace{-1mm}+\hspace{-1mm}\frac{N_0\big(P_{\mathcal{S}_1}|h_{\mathcal{S}_{1}\mathcal{J}}|^2+P_{S_{2}}|h_{\mathcal{S}_{2}\mathcal{J}}|^2\big)|h_{\mathcal{JR}}|^2}{P_{\mathcal{S}_{1}}|h_{\mathcal{S}_1\mathcal{R}}|^2+{P_{\mathcal{S}_{2}}|h_{\mathcal{S}_2\mathcal{R}}|^2}}\hspace{-1mm}+\hspace{-1mm}N_0\beta\hspace{-1mm}+\hspace{-1mm}\epsilon}.
\end{align}

\subsubsection{Gaussian Noise Jamming}
For GNJ, the received SNR at~\rr~ is the same as the FJ. However, in this scenario, since the jamming signal is not available at both~\src~and~\des, the term related to the jamming signal $x_{\mathcal{J}}$ in \eqref{ys2prime_fj} is considered as a noise-like interference. Consequently, after self-interference cancellation, the received signal-to-interference-plus-noise ratio (SINR) at~\des~can be computed as
\begin{align} \label{gammad1_gnj}
\gamma_{\mathcal{S}_2}\hspace{-1mm}=\hspace{-1mm}\frac{{P_{\mathcal{S}_{1}}}|h_{\mathcal{S}_1\mathcal{R}}|^2 |h_{\mathcal{RS}_{2}}|^2}{
\splitfrac{{\Big(\hspace{-1mm}P_{\mathcal{S}_1}|h_{\mathcal{S}_{1}\mathcal{J}}|^2\hspace{-1mm}+P_{S_{2}}|h_{\mathcal{S}_{2}\mathcal{J}}|^2\Big){\Big(\beta^{-1}|h_{\mathcal{R}\mathcal{S}_{2}}|^2\hspace{-1mm}+\delta\Big)}|h_{\mathcal{J}\mathcal{R}}|^2}}{{+N_{0}|h_{\mathcal{RS}_{2}}|^2\hspace{-1mm}+\hspace{-1mm}N_0\beta\hspace{-1mm}+\hspace{-1mm}\epsilon}}},
\end{align}
where $\delta\treq\frac{N_0}{P_{\mathcal{S}_{1}}|h_{\mathcal{S}_1\mathcal{R}}|^2+{P_{\mathcal{S}_{2}}|h_{\mathcal{S}_2\mathcal{R}}|^2}}$. A similar expression can be obtained for $\gamma_{\mathcal{S}_1}$ by changing ${\mathcal{S}_2}$ with ${\mathcal{S}_1}$ in  \eqref{gammad1_gnj}.

To make the further analysis tractable, we consider the high SNR assumption for all the scenarios by replacing $\epsilon=0$ in Eqs. (\ref{gammad1_woj}), (\ref{gammad2_woj}) and (\ref{gammad1_fj})-\eqref{gammad1_gnj}.

\section{Ergodic Secrecy Sum Rate Analysis}

In this section, we first derive closed-form expressions for the power outage probability at the helping nodes to take into account the fact that the EH may fail at either~\rr~or~\jj. Then, we analytically obtain new closed-form lower-bound expressions for the ESSR of WoJ, FJ, and GNJ. 

We assume the helping nodes only utilize the wireless EH for data transmission. As such, the received power at either~\rr~or~\jj,~should be greater than the minimum required power for the activation of their EH circuitry \cite{Ku}, unless they maintain inactive as we assume the helping nodes only utilize the wireless EH technology and have no other power resources. This phenomenon is characterized by the power outage probability, and denoted by $P_{po}$. In this section, we first derive closed-form expressions for the power outage probability at~\rr~($P_{po}^{\mathcal{R}}$), and~\jj~($P_{po}^{\mathcal{J}}$). As such, the probability of power outage for the helper node $\mathcal{K}$, where $\mathcal{K}\in$ \{\rr,~\jj\} is defined precisely as
\begin{equation}\label{pop}
P_{po}^{\mathcal{K}}=\Pr\{P_{\mathcal{K}} < \Theta \},
\end{equation}
in which the analytical expression for $P_{po}^{\mathcal{K}}$ is obtained in Proposition \ref{prop1}.
\begin{proposition}\label{prop1}
	The power outage probability at the helper node $\mathcal{K}$, where $\mathcal{K}\in$ \{\rr,~\jj\} is given by
	\begin{align}\label{por}
	P_{po}^{\mathcal{K}}=\begin{cases} 
	1-\frac{\bar{\gamma}_{\mathcal{S}_{2}\mathcal{K}}}{\bar{\gamma}_{\mathcal{S}_{2}\mathcal{K}}-\bar{\gamma}_{\mathcal{S}_{1}\mathcal{K}}}\exp(-\frac{\Theta}{\bar{\gamma}_{\mathcal{S}_{2}\mathcal{K}}})\\~-\frac{\bar{\gamma}_{\mathcal{S}_{1}\mathcal{K}}}{\bar{\gamma}_{\mathcal{S}_{1}\mathcal{K}}-\bar{\gamma}_{\mathcal{S}_{2}\mathcal{K}}}	\exp(-\frac{\Theta}{\bar{\gamma}_{\mathcal{S}_{1}\mathcal{K}}})
	,& \bar{\gamma}_{\mathcal{S}_{1}\mathcal{K}} \neq \bar{\gamma}_{\mathcal{S}_{2}\mathcal{K}}\\ \\  
	\Upsilon(2,\frac{\Theta}{\bar{\gamma}_{\mathcal{S}_{1}\mathcal{K}}}),& \bar{\gamma}_{\mathcal{S}_{1}\mathcal{K}} =\bar{\gamma}_{\mathcal{S}_{2}\mathcal{K}}
	\end{cases}
	\end{align}
	where $\bar{\gamma}_{\mathcal{S}_{1}\mathcal{K}}\treq{P_{\mathcal{S}_{1}}}\mu_{\mathcal{S}_1\mathcal{K}}$, $\bar{\gamma}_{\mathcal{S}_{2}\mathcal{K}}\treq{P_{\mathcal{S}_{2}}}\mu_{\mathcal{S}_2\mathcal{K}}$, and $\Upsilon(s,x){\hspace {-1mm}}={\hspace {-1mm}}\int_{0}^{x}t^{(s-1)}\emph{e}^{-t} dt$ is the lower incomplete Gamma function \cite{papoulis}.
\end{proposition}
\begin{proof}
	See Appendix A.
\end{proof} 

In principle, the ergodic secrecy rate determines the rate below which any average secure transmission is accessible {\cite {review0}}. Since we assume the MUD is performed at the untrusted relay to decode both the signals $x_{\mathcal{S}_1}$ and $x_{\mathcal{S}_2}$, the integrated secrecy rate of the communication network, is considered as \cite{zhang0}. Therefore, the instantaneous secrecy sum rate $R_{Sec}$ is evaluated by
\begin{equation}\label{rsec1}
R_{Sec}=\left[I_{\mathcal{S}_1}+I_{\mathcal{S}_2}-I_{\mathcal{R}}\right]^+,
\end{equation}
where for ${K}\in$\{\src, \rr, \des\}
\begin{equation}\label{rsec2}
I_{{{K}}}=\frac{(1-\alpha)}{2}\log_{2}(1+\gamma_{{K}}),
\end{equation}
By combining (\ref{rsec1}) and (\ref{rsec2}), $R_{Sec}$ can be rewritten as
\begin{equation}\label{issr}
R_{Sec}=\left[\frac{(1-\alpha)}{2}\log_{2}\frac{(1+\gamma_{\mathcal{S}_1})(1+\gamma_{\mathcal{S}_2})}{(1+\gamma_{\mathcal{R}})}\right]^+,
\end{equation}
where $[x]^+ = \max(x,0)$ and the pre-log factor  $\frac{1-\alpha}{2}$ is due to the efficient time of information exchange between the two sources. Moreover, $\gamma_{\mathcal{S}_1}$, $\gamma_{\mathcal{S}_2}$, and  $\gamma_{\mathcal{R}}$ are the received SNR at~\src,~\des, and~\rr, respectively. We note that by taking average over ${{R}_{Sec}}$ given by \eqref{issr}, one can obtain the ESSR as
\begin{equation}\label{essrformula}
{\bar{R}_{Sec}}=\E\{R_{Sec}\}.
\end{equation}

In the following, we proceed to derive the ESSR of the WoJ, FJ, and GNJ scenarios.
\subsection{Without Jamming}
In this scenario,~\rr~may experience power outage due to bad channel conditions. Hence, the ESSR of WoJ can be stated as
\begin{equation}\label{two-wayWoFJ_Exact}
\bar{R}_{Sec}^{WoJ}=(1-P_{po}^{\mathcal{R}})\bar{R}_{Act}^{WoJ},
\end{equation}
where the exact expression of $\bar{R}_{Act}^{WoJ}$  is obtained by substituting \eqref{gammaR_woj}, \eqref{gammad1_woj}, and \eqref{gammad2_woj} into \eqref{essrformula} as
\begin{align}\label{Rsecwoj}
\bar{R}_{Act}^{WoJ}=\int_{0}^{\infty}\int_{0}^{\infty}R_{sec}(x,y)f_{X}(x)f_{Y}(y)dxdy,
\end{align}
where $X{\hspace{-1mm}}={\hspace{-1mm}}|h_{\mathcal{S}_1\mathcal{R}}|^2$ and $Y{\hspace{-1mm}}={\hspace{-1mm}}|h_{\mathcal{S}_2\mathcal{R}}|^2$ are defined in \eqref{Rsecwoj}.

The corresponding lower-bound expression for $\bar{R}_{Act}^{WoJ}$ can be analytically formulated as
\begin{equation}\label{two-wayWoFJ_LB}
\bar{R}_{LB}^{WoJ}=\frac{1-\alpha}{2\ln(2)}\left[\widehat{I}_{1}+\widehat{I}_{2}-I_3\right]^+,
\end{equation}
where 
\begin{align}
\widehat{I}_{1(2)}&\hspace{-1mm}=\hspace{-1mm}\ln\Bigg[1+\exp\Bigg(
 -2\Phi+\ln\left(\frac{\bar{\gamma}_{\mathcal{S}_{1(2)}\mathcal{R}}\mu_{{\mathcal{R}\mathcal{S}_{2(1)}}}}{\beta N_0}\right) \nonumber\\ &\hspace{16mm}+\exp\big(\frac{\beta}{\mu_{\mathcal{R}\mathcal{S}_{2(1)}}}\big){\mathrm{Ei}}\big(-\frac{\beta}{\mu_{\mathcal{R}\mathcal{S}_{2(1)}}}\big)\Bigg)\Bigg],
\end{align}
where $\Phi\hspace{-1mm}\approx\hspace{-1mm}0.577215$ is the Euler's constant \cite{integ}, and ${ \mathrm{Ei}}(x){\hspace{-1mm}}={\hspace{-1mm}}-\int_{-x}^{\infty} \frac{\exp(-t)}{t}dt$ is the exponential integral \cite{papoulis}. Furthermore, the term $I_3$ is given by
\begin{align}
I_3&=\frac{\bar{\gamma}_{\mathcal{S}_{1}\mathcal{R}}}{\bar{\gamma}_{\mathcal{S}_{2}\mathcal{R}}-\bar{\gamma}_{\mathcal{S}_{1}\mathcal{R}}}\exp\left(\frac{N_0}{\bar{\gamma}_{\mathcal{S}_{1}\mathcal{R}}}\right)\mathrm{Ei}\left(-\frac{N_0}{\bar{\gamma}_{\mathcal{S}_{1}\mathcal{R}}}\right)\nonumber\\
&+\frac{\bar{\gamma}_{\mathcal{S}_{2}\mathcal{R}}}{\bar{\gamma}_{\mathcal{S}_{1}\mathcal{R}}-\bar{\gamma}_{\mathcal{S}_{2}\mathcal{R}}}\exp\left(\frac{N_0}{\bar{\gamma}_{\mathcal{S}_{2}\mathcal{R}}}\right)\mathrm{Ei}\left(-\frac{N_0}{\bar{\gamma}_{\mathcal{S}_{2}\mathcal{R}}}\right).
\end{align}
\begin{proof}
See Appendix B.
\end{proof}

\subsection{Friendly Jamming}

By considering this fact that  the power outage may occur at either~\rr~or~\jj, the ESSR for FJ can be written as
\begin{equation}\label{pintfj-total}
\bar{R}_{Sec}^{FJ}=P_{po}^{\mathcal{J}}\bar{R}_{Sec}^{WoJ}+(1-P_{po}^{\mathcal{R}})(1-P_{po}^{\mathcal{J}})\bar{R}_{Act}^{FJ}.
\end{equation}
We mention that the exact ESSR expression for FJ assuming all the nodes are active, $\bar{R}_{Act}^{FJ}$, can be written as
\begin{align}\label{essr}
\bar{R}_{Act}^{FJ}&=\int_{0}^{\infty}\int_{0}^{\infty}\int_{0}^{\infty}\int_{0}^{\infty}\int_{0}^{\infty} R_{sec}(x,y,z,w,u) \nonumber \\
&\times f_{X}(x)f_{Y}(y)f_{Z}(z)f_{U}(u)f_{W}(w)dxdydzdudw,
\end{align}
where we define $X{\hspace{-1mm}}={\hspace{-1mm}}|h_{\mathcal{S}_1\mathcal{R}}|^2$, $Y{\hspace{-1mm}}={\hspace{-1mm}}|h_{\mathcal{S}_2\mathcal{R}}|^2$, $Z{\hspace{-1mm}}={\hspace{-1mm}}|h_{\mathcal{S}_1\mathcal{J}}|^2$, $W{\hspace{-1mm}}={\hspace{-1mm}}|h_{\mathcal{S}_2\mathcal{J}}|^2$, and $U{\hspace{-1mm}}={\hspace{-1mm}}|h_{\mathcal{RJ}}|^2$ in  the RVs of $\gamma_{\mathcal{R}}$, $\gamma_{\mathcal{S}_2}$, and $\gamma_{\mathcal{S}_1}$, which are respectively given by  \eqref{gammaR_fj}, \eqref{gammad1_fj}, and \eqref{gammad2_fj}.

Although the multiple integral expression in \eqref{essr} can be evaluated numerically, a closed-form expression is not straightforward to obtain. As such, we proceed by deriving a new compact lower-bound expression for $\bar{R}_{Act}^{FJ}$ in Proposition \ref{prop2}.

\begin{proposition}\label{prop2}
The lower-bound expression for the ESSR of FJ scenario when both the helpers maintain active ($\bar{R}_{LB}^{FJ}$) can be expressed as 
\begin{align}\label{R_LB}
\bar{R}_{LB}^{FJ}= \frac{1-\alpha}{2\ln(2)}\left[\mathcal{L}_1+\mathcal{L}_2-\mathcal{L}_3\right]^+,
\end{align}
where
\begin{align} \label{L_1}
\mathcal{L}_1&\geq\ln{\hspace{-1.5mm}}\left({\hspace{-1mm}}1{\hspace{-1mm}}+{\hspace{-1mm}}\frac{\exp\Big[-2\Phi{\hspace{-1mm}}+{\hspace{-1mm}}\ln\big(P_{\mathcal{S}_{1}}\mu_{\mathcal{S}_1\mathcal{R}}\mu_{\mathcal{S}_2\mathcal{R}}\big)\Big]}{N_{0}\Big[\mu_{\mathcal{R}\mathcal{S}_2}{\hspace{-1mm}}+{\hspace{-1mm}}\beta{\hspace{-1mm}}+{\hspace{-1mm}}\mu_{\mathcal{JR}}\frac{{P_{\mathcal{S}_{1}}}\mu_{\mathcal{S}_1\mathcal{J}}+{P_{\mathcal{S}_{2}}}\mu_{\mathcal{S}_2\mathcal{J}}}{{P_{\mathcal{S}_{1}}}\mu_{\mathcal{S}_1\mathcal{R}}-{P_{\mathcal{S}_{2}}}\mu_{\mathcal{S}_2\mathcal{R}}}\ln \frac{{P_{\mathcal{S}_{1}}}\mu_{\mathcal{S}_1\mathcal{R}}}{{P_{\mathcal{S}_{2}}}\mu_{\mathcal{S}_2\mathcal{R}}}\Big]}\right). \nonumber\\
\end{align}
\begin{align}\label{L_2}
\mathcal{L}_2&\geq\ln{\hspace{-1.5mm}}\left({\hspace{-1mm}}1{\hspace{-1mm}}+{\hspace{-1mm}}\frac{\exp\Big[-2\Phi{\hspace{-1mm}}+{\hspace{-1mm}}\ln\big(P_{\mathcal{S}_{2}}\mu_{\mathcal{S}_2\mathcal{R}}\mu_{\mathcal{S}_1\mathcal{R}}\big)\Big]}{N_{0}\Big[\mu_{\mathcal{R}\mathcal{S}_1}{\hspace{-1mm}}+{\hspace{-1mm}}\beta{\hspace{-1mm}}+{\hspace{-1mm}}\mu_{\mathcal{JR}}\frac{{P_{\mathcal{S}_{1}}}\mu_{\mathcal{S}_1\mathcal{J}}+{P_{\mathcal{S}_{2}}}\mu_{\mathcal{S}_2\mathcal{J}}}{{P_{\mathcal{S}_{1}}}\mu_{\mathcal{S}_1\mathcal{R}}-{P_{\mathcal{S}_{2}}}\mu_{\mathcal{S}_2\mathcal{R}}}\ln \frac{{P_{\mathcal{S}_{1}}}\mu_{\mathcal{S}_1\mathcal{R}}}{{P_{\mathcal{S}_{2}}}\mu_{\mathcal{S}_2\mathcal{R}}}\Big]}\right). \nonumber\\
\end{align}
and
\begin{equation}\label{L_3}
\mathcal{L}_3\leq\ln\left(1+\mathcal{A}_0\Big[\mathcal{F}(\mathcal{A}_1)-\mathcal{F}(\mathcal{A}_2)\Big]\right)
\end{equation}
where 
\begin{equation}
\mathcal{A}_0=\frac{2\beta\left({P_{\mathcal{S}_{1}}}\mu_{\mathcal{S}_1\mathcal{R}}+{P_{\mathcal{S}_{2}}}\mu_{\mathcal{S}_2\mathcal{R}}\right)}{({P_{\mathcal{S}_{2}}}\mu_{\mathcal{S}_2\mathcal{J}}-{P_{\mathcal{S}_{1}}}\mu_{\mathcal{S}_1\mathcal{J}})\mu_{\mathcal{RJ}}},
\end{equation}
and
\begin{equation}
\mathcal{A}_1=\sqrt{\frac{4\beta N_0}{{{P_{\mathcal{S}_{1}}}\mu_{\mathcal{S}_1\mathcal{J}}}\mu_{\mathcal{RJ}}}},~~~~
\mathcal{A}_2=\sqrt{\frac{4\beta N_0}{{{P_{\mathcal{S}_{2}}}\mu_{\mathcal{S}_2\mathcal{J}}}\mu_{\mathcal{RJ}}}},
\end{equation}
and, also for $m \in \{1, 2\}$
\begin{align}
\mathcal{F}(\mathcal{A}_m)&=-2\sum_{n=1}^{\infty}\sum_{i=1}^{n}\Lambda(1,n,i)\bigg(\frac{9}{2}
\frac{\Gamma(n-\frac{3}{4})\Gamma(n+\frac{3}{2})}{\Gamma(n-\frac{1}{2})\Gamma(n+\frac{5}{2})}+2\bigg)\nonumber\\
&\hspace*{-12mm}\times\mathcal{A}_m^{i-2}\begin{cases}
\bigg((-1)^k\Big[\mathrm{ci}(\mathcal{A}_m)\cos(\mathcal{A}_m)+\mathrm{si}(\mathcal{A}_m)\sin(\mathcal{A}_m)\Big]\nonumber\\
+\frac{1}{\mathcal{A}_m^{2k-2}}\sum\limits_{j=1}^{k-1}(2k-2j-1)!(-\mathcal{A}_m^2)^{j-1}\bigg),~~~i=2k\\\\
\bigg((-1)^k\Big[\mathrm{ci}(\mathcal{A}_m)\sin(\mathcal{A}_m)-\mathrm{si}(\mathcal{A}_m)\cos(\mathcal{A}_m)\Big]\nonumber\\
+\frac{1}{\mathcal{A}_m^{2k-1}}\sum\limits_{j=1}^{k}(2k-2j)!(-\mathcal{A}_m^2)^{2j-1}\bigg),~~~i=2k+1 
\end{cases}\\
\end{align}
where
\begin{equation}
\Lambda(1,n,i)=-{\frac { \left( -2 \right) ^{i}\sqrt {\pi}\mathrm{L}(i,n)}{\sqrt {\pi }
\Gamma  \left( n+1 \right)  \left( 4\,{n}^{2}-1 \right) }},
\end{equation}
where $\mathrm{L}(i,n)={{{n-1\choose i-1}}\frac{n!}{i!}}$ for $n, i >0$ represents the Lah numbers (e.g. \cite{LahNum}), $\Gamma(\cdot)$ is Gamma function. Also, $\mathrm{ci}(x)$ and $\mathrm{si}(x)$ are the Sine and Cosine integrals, i.e., $\mathrm{si}(x)=-\int_{x}^{\infty}\frac{\sin(t)}{t}\mathrm{d}t$ and $\mathrm{ci}(x)=-\int_{x}^{\infty}\frac{\cos(t)}{t}\mathrm{d}t$, respectively.

\end{proposition}
\begin{proof}
See Appendix C.
\end{proof}

As shown in the numerical results, the novel lower-bound expression given by (\ref{R_LB}) is significantly tight, especially in the moderate-to-high SNR regime.

\subsection{Gaussian Noise Jamming}

The ESSR of GNJ scenario can be obtained following the same procedure done for FJ scenario. We must only add the term $\beta^{-1}\big(P_{\mathcal{S}_{1}}\mu_{\mathcal{S}_1\mathcal{J}}+P_{\mathcal{S}_{2}}\mu_{\mathcal{S}_2\mathcal{J}}\big)\mu_{\mathcal{R}\mathcal{S}_i}\mu_{\mathcal{J}\mathcal{R}}$, for $i\in\{1, 2\}$, to the denominator of rational functions in \eqref{L_1} and \eqref{L_2}, respectively.


\section{Asymptotic Ergodic Secrecy Sum Rate Analysis}

In this section, we obtain the asymptotic ESSR when the transmit SNR of each node goes to infinity by deriving the high SNR slope in bits/s/Hz ($S_{\infty}$) and the high SNR power offset in 3 dB units ($L_{\infty}$), which are defined respectively as
\begin{equation}\label{SL}
S_{\infty}=\lim_{\rho\to\infty}\frac{\bar{R}_{Sec}^\infty}{\log_{2}\rho}~~\mathrm{and}~~ L_{\infty}=\lim_{\rho\to\infty}\big({\log_{2}\rho}-\frac{\bar{R}_{Sec}^\infty}{S_{\infty}}\big),
\end{equation}
where
\begin{equation}
\bar{R}_{Sec}^{\infty}=S_{\infty}(\log_{2}\rho-L_{\infty}),
\end{equation}
is the general asymptotic form of the ESSR performance \cite{Lozano}.

For the ease of presentation, we assume that $P_{\mathcal{S}_1}$ and $P_{\mathcal{S}_2}$ grow large with $P_{\mathcal{S}_1}=\xi P_{\mathcal{S}_2}$ for some fixed ratio $0<\xi<\infty$. Furthermore, we define $\rho=\frac{P_{\mathcal{S}_2}}{N_0}$ as the transmit SNR by~\des.

\subsubsection{Without Jamming}

In the high SNR regime with $\rho\rightarrow\infty$ and based on (\ref{gammaR_woj}), (\ref{gammad1_woj}), and  (\ref{gammad2_woj}), we conclude that $\ln (1+\gamma_{\mathcal{S}_\textit{i}})\approx \ln (\gamma_{\mathcal{S}_\textit{i}})$ for $\textit{i}\in\{1, 2\}$, and $\ln (1+\gamma_\mathcal{R})\approx \ln (\gamma_\mathcal{R})$. As such,
\begin{eqnarray}\label{Y1}
{\mathcal{Y}_1}{\hspace{-3mm}}&\approx&{\hspace{-3mm}}\E\{\ln(\gamma_{\mathcal{S}_1})\}={\E\Big\{\ln(\frac{\xi\rho XY}{X+\beta})\Big\}}\nonumber\\
{\hspace{-3mm}}&=&{\hspace{-3mm}}\ln(\xi\rho)+\E\Big\{\ln(XY)\Big\}-\E\Big\{\ln(X+\beta)\Big\},
\end{eqnarray}
\begin{eqnarray}\label{Y2}
{\mathcal{Y}_2}{\hspace{-3mm}}&\approx&{\hspace{-3mm}}\E\{\ln(\gamma_{\mathcal{S}_2})\}={\E\Big\{\ln(\frac{\rho XY}{Y+\beta})\Big\}}\nonumber\\
{\hspace{-3mm}}&=&{\hspace{-3mm}}\ln(\rho)+\E\Big\{\ln(XY)\Big\}-\E\Big\{\ln(Y+\beta)\Big\},
\end{eqnarray}
\begin{eqnarray}\label{Y3}
{\mathcal{Y}_3}{\hspace{-3mm}}&\approx&{\hspace{-3mm}}\E\{\ln(\gamma_{\mathcal{R}})\}={\E\Big\{\ln(\xi\rho X+\rho Y)\Big\}}\nonumber\\
{\hspace{-3mm}}&=&{\hspace{-3mm}}\ln(\xi\rho)+\E\Big\{\ln(X+\frac{1}{\xi}Y)\Big\},
\end{eqnarray}
where the terms $\E\{\ln (XY)\}$, $\E\{\ln(X+C)\}$ and $\E\{\ln (X+CY)\}$ can be evaluated using the lemma mentioned below.

\begin{lemma}\label{lem}
Let $C$  be a strictly positive constant, and $X$ and $Y$  be two different exponential RVs with means of $m_x$ and $m_y$, respectively. Therefore, we have the following results.
\begin{eqnarray}
1)~~~&&{\hspace{-7mm}}\E\big\{\ln X\big\}=\ln(m_x)-\Phi,\nonumber\\
2)~~~&&{\hspace{-7mm}}\E\big\{\ln(X+C)\big\}=\ln(C)-e^\frac{C}{m_x}\mathrm{Ei}(-\frac{C}{m_x})\nonumber\\
3)~~~&&{\hspace{-7mm}}\E\big\{\ln(X+CY)\big\}{\hspace{-1mm}}={\hspace{-1mm}}\frac{Cm_xm_y}{m_x-Cm_y}{\hspace{-0.5mm}} \nonumber\\	&&~~~~~~~~~~~~~~~\times\bigg[\frac{\Phi+\ln(m_x)}{m_x}{\hspace{-1mm}}-{\hspace{-1mm}}\frac{\Phi+\ln(Cm_y)}{Cm_y}{\hspace{-0.5mm}}\bigg]. \nonumber
\end{eqnarray}
\end{lemma}
\begin{proof}
This lemma can be proved using {\cite[Eq. (4.331.1)] {integ}} for expression 1, using {\cite[Eq. (4.337.1)] {integ}} for expression 2, and using {\cite[Eq. (4.352.2)] {integ}} for expression 3.
\end{proof}
Applying Lemma \ref{lem} to \eqref{Y1}, \eqref{Y2}, and \eqref{Y3}, and then substituting them into \eqref{essrformula}, the closed-form expression for the asymptotic ESSR of the WoJ, can be obtained as
\begin{align}\label{RWoFJin}
\bar{R}_{Sec}^{WoJ,\infty}&=\big(1-P_{po}^{\mathcal{R}}\big)\bar{R}_{Act}^{WoJ,\infty}\nonumber\\
&=\big(1-P_{po}^{\mathcal{R}}\big)\frac{1-\alpha}{2\ln 2}\Bigg(\ln(\rho)+2\ln\Big(\frac{m_xm_y}{\beta}\Big)\nonumber\\
&-4\Phi+e^{\frac{\beta}{m_x}}{\mathrm{Ei}}\Big(-\frac{\beta}{m_x}\Big)+e^{\frac{\beta}{m_y}}\mathrm{Ei}\Big(-\frac{\beta}{m_y}\Big)\nonumber\\
&-\frac{m_xm_y}{\xi m_x-m_y}\bigg[\frac{\Phi+\ln(m_x)}{m_x}-\frac{\xi\Phi+\xi\ln(\frac{m_y}{\xi} )}{m_y}\bigg]\Bigg).
\end{align}
By substituting \eqref{RWoFJin} into \eqref{SL}, we arrive at the high SNR slope and the high SNR power offset respectively, as
\begin{equation}\label{SwoFJ}
S_{\infty}^{WoJ}=(1-P_{po}^{\mathcal{R}})\frac{1-\alpha}{2}.
\end{equation}
and
\begin{align}\label{Lwoj}
L_{\infty}^{WoJ}&=\frac{1}{\ln 2}\Bigg(4\Phi-e^{\frac{\beta}{m_x}}{\mathrm{Ei}}(-\frac{\beta}{m_x})-e^{\frac{\beta}{m_y}}\mathrm{Ei}(-\frac{\beta}{m_y})\nonumber\\
&+\frac{m_xm_y}{\xi m_x-m_y}\Big[\frac{\Phi+\ln(m_x)}{m_x}-\frac{\xi\Phi+\xi\ln(\frac{m_y}{\xi} )}{m_y}\Big]\Bigg)\nonumber\\
&-2\log_2(\frac{m_xm_y}{\beta}).
\end{align}

\subsubsection{Friendly Jamming}
The asymptotic ESSR for the FJ scenario becomes as
\begin{align}\label{abbb}
\bar{R}_{Sec}^{FJ,\infty}&=P_{po}^{\mathcal{J}}\bar{R}_{Sec}^{WoJ,\infty}+(1-P_{po}^{\mathcal{R}})(1-P_{po}^{\mathcal{J}})\bar{R}_{Act}^{FJ,\infty},
\end{align}
where $\bar{R}_{Act,\infty}^{FJ}$ in \eqref{abbb}, can be expressed as
\begin{align}\label{Rwfj}
\bar{R}_{Act}^{FJ,\infty}\hspace{-1mm}=\hspace{-1mm}\frac{1-\alpha}{2\ln 2}\Bigg[\hspace{-1mm}\stackrel{}{\underset{\mathcal{J}_1}{\underbrace{\E\Big\{\ln(\gamma_{\mathcal{S}_1})\Big\}}}}\hspace{-1mm}+\hspace{-1mm}
\stackrel{}{\underset{\mathcal{J}_2}{\underbrace{\E\Big\{\ln(\gamma_{\mathcal{S}_2})\Big\}}}}\hspace{-1mm}-\hspace{-1mm}\stackrel{}{\underset{\mathcal{J}_3}{\underbrace{\E\Big\{\ln(\gamma_{R})\Big\}}}}\hspace{-1mm}\Bigg], 
\end{align}
where using (\ref{gammad2_fj}), (\ref{gammad1_fj}) and (\ref{gammaR_fj}), the terms $\mathcal{J}_1$, $\mathcal{J}_2$ and $\mathcal{J}_3$ are derived as follows:
\begin{eqnarray}\label{J1}
{\mathcal{J}_1}{\hspace{-3mm}}&=&{\hspace{-3mm}}{\E\Bigg\{\ln\bigg(\frac{\rho XY}{X+\frac{\xi Z+W}{\xi X+Y}U+\beta}\bigg)\Bigg\}}\nonumber\\
{\hspace{-3mm}}&=&{\hspace{-3mm}}\E\Big\{\ln(\rho XY)\Big\}-\E\Big\{\ln\Big({X+\frac{\xi Z+W}{\xi X+Y}U+\beta}\Big)\Big\}\nonumber\\
{\hspace{-3mm}}&\stackrel{(a)}{\geq}&{\hspace{-3mm}}\ln(\rho)+\E\Big\{\ln( XY)\Big\}-\ln\bigg(\E\Big\{{X+\frac{\xi Z+W}{\xi X+Y}U+\beta}\Big\}\bigg)\nonumber\\
{\hspace{-3mm}}&\stackrel{(b)}{=}&{\hspace{-3mm}}\ln(\rho)+\ln(m_xm_y)-2\Phi\nonumber\\
{\hspace{-3mm}}&-&{\hspace{-3mm}}\ln\Big[\beta+m_x+\frac{\xi m_z+m_w}{\xi m_x-m_y}m_u\ln(\frac{\xi m_x}{m_y})\Big],
\end{eqnarray}
where $(a)$ follows from Jensen's inequality, and $(b)$ follows from using Lemma \ref{lem}. Similar to $\mathcal{J}_1$, we obtain $\mathcal{J}_2$ as 
\begin{eqnarray}\label{J2}
{\mathcal{J}_2}{\hspace{-3mm}}&\geq&{\hspace{-3mm}}\ln(\rho)+\ln(\xi m_xm_y)-2\Phi\nonumber\\
{\hspace{-3mm}}&-&{\hspace{-3mm}}\ln\Big[\beta+m_y+\frac{\xi m_z+m_w}{\xi m_x-m_y}m_u\ln(\frac{\xi m_x}{m_y})\Big].
\end{eqnarray}
Ultimately, the term $\mathcal{J}_3$ is derived as 
\begin{eqnarray}\label{J3}
{\mathcal{J}_3}{\hspace{-3mm}}&=&{\hspace{-3mm}}{\E\Bigg\{\ln\bigg(\frac{X+\frac{1}{\xi}Y}{\frac{1}{\beta}(Z+ \frac{1}{\xi}W)U+\epsilon}\bigg)\Bigg\}}\nonumber\\
{\hspace{-3mm}}&\stackrel{(a)}{\approx}&{\hspace{-3mm}}\E\Big\{\ln(X+\frac{1}{\xi}Y)\Big\}-\E\Big\{\ln(Z+\frac{1}{\xi}W)\Big\}\nonumber\\ 
{\hspace{-3mm}}&-&{\hspace{-3mm}}\E\Big\{\ln(U)\Big\}+\ln \beta\nonumber\\
{\hspace{-3mm}}&\stackrel{(b)}{=}&{\hspace{-3mm}}\frac{\frac{1}{\xi}m_xm_y}{m_x-\frac{1}{\xi}m_y}\Big[\frac{\Phi+\ln(m_x)}{m_x}-\frac{\Phi+\ln(\frac{m_y}{\xi})}{\frac{1}{\xi}m_y}\Big]\nonumber\\
{\hspace{-3mm}}&-&{\hspace{-3mm}}\frac{\frac{1}{\xi}m_zm_w}{m_x-\frac{1}{\xi}m_w}\Big[\frac{\Phi+\ln(m_z)}{m_z}-\frac{\Phi+\ln(\frac{m_w}{\xi})}{\frac{1}{\xi}m_w}\Big]\nonumber\\
{\hspace{-3mm}}&+&{\hspace{-3mm}}\ln\big(\frac{\beta}{m_u}\big)+\Phi,
\end{eqnarray}
where $(a)$ follows from setting $\epsilon=0$; this means that the untrusted relay is considered as an ideal eavesdropper with the capability of noise cancellation such that from a security perspective this corresponds to the maximum interception by the eavesdropper and is the worst case assumption \cite{vucetic2009}. Furthermore, $(b)$ follows from Lemma \ref{lem}. Consequently, substituting (\ref{J1})-(\ref{J3}) into (\ref{Rwfj}), and then using \eqref{abbb} and \eqref{SL}, the high SNR slope for the FJ, is given by
\begin{equation}\label{sinffj}
S_{\infty}^{FJ}=P_{po}^{\mathcal{J}}S_{\infty}^{WoJ}+(1-P_{po}^{\mathcal{R}})(1-P_{po}^{\mathcal{J}})S_{\infty}^{FJ, Act},
\end{equation}
where by plugging \eqref{Rwfj} into \eqref{SL}, the expression $S_{\infty}^{FJ, Act}$ can be expressed as
\begin{equation}\label{SWFJ}
S_{\infty}^{FJ, Act}=(1-\alpha).
\end{equation}
Ultimately, substituting \eqref{SWFJ} into \eqref{sinffj}, and then after simple manipulations results in
\begin{equation}\label{finalSinfFJ}
S_{\infty}^{FJ}=(1-P_{po}^{\mathcal{R}})(1-\frac{P_{po}^{\mathcal{J}}}{2})(1-\alpha).
\end{equation}
Finally, for the calculation of the high SNR power offset for the FJ, plugging \eqref{abbb} into \eqref{SL} results in
\begin{align}
\hspace{-3mm}L_{\infty}^{FJ}\hspace{-1.5mm}=\hspace{-2mm}\lim_{\rho\to\infty}\hspace{-1mm}\Bigg(\hspace{-0.5mm}{\log_{2}\rho}\hspace{-0.5mm}-\hspace{-1mm}\Bigg[\frac{P_{po}^{\mathcal{J}}\bar{R}_{Act}^{WoJ,\infty}\hspace{-0.5mm}+\hspace{-0.5mm}(1-P_{po}^{\mathcal{J}})\bar{R}_{Act}^{FJ,\infty}}{(1-\frac{P_{po}^{\mathcal{J}}}{2})(1-\alpha)}
\hspace{-0.5mm}\Bigg]\hspace{-0.5mm}\Bigg).
\end{align}
Now, we consider two special cases 1) jammer is always active, i.e., ${P_{po}^{\mathcal{J}}}=0$, which is an ideal case maximizing $L_{\infty}^{FJ}$ , 2) Jammer is off, ${P_{po}^{\mathcal{J}}}=1$, which is also an artificial case but minimizing $L_{\infty}^{FJ}$. we delve into such computations to acquire a deep engineering insight to these criterion. To this end, if ${P_{po}^{\mathcal{J}}}=0$, then
\begin{align}
L_{\infty}^{FJ, Act}&\hspace{-1mm}=\hspace{-1mm}\frac{1}{2\ln 2}\Bigg(\ln\frac{\beta}{\xi m_x^2m_y^2m_u}+5\Phi\nonumber\\
&\hspace{-1mm}+\ln\Big[\beta+m_x+\frac{\xi m_z+m_w}{\xi m_x\hspace{-0.5mm}-\hspace{-0.5mm}m_y}m_u\ln(\frac{\xi m_x}{m_y})\Big]\nonumber\\
&\hspace{-1mm}+\ln\Big[\beta+m_y+\frac{\xi m_z+m_w}{\xi m_x-m_y}m_u\ln(\frac{\xi m_x}{m_y})\Big]\nonumber\\
&\hspace{-1mm}+\hspace{-1mm}\frac{\frac{1}{\xi}m_xm_y}{m_x\hspace{-0.5mm}-\hspace{-0.5mm}\frac{1}{\xi}m_y}\Big[\frac{\Phi\hspace{-0.5mm}+\hspace{-0.5mm}\ln(m_x)}{m_x}\hspace{-1mm}-\hspace{-1mm}\frac{\Phi\hspace{-0.5mm}+\hspace{-0.5mm}\ln(\frac{m_y}{\xi})}{\frac{1}{\xi}m_y}\Big]\nonumber\\
&\hspace{-1mm}-\hspace{-1mm}\frac{\frac{1}{\xi}m_zm_w}{m_z\hspace{-0.5mm}-\hspace{-0.5mm}\frac{1}{\xi}m_w}\hspace{-0.5mm}\Big[\hspace{-0.5mm}\frac{\Phi+\ln(m_z)}{m_z}\hspace{-1mm}-\hspace{-1mm}\frac{\Phi\hspace{-0.5mm}+\hspace{-0.5mm}\ln(\frac{m_w}{\xi})}{\frac{1}{\xi}m_w}\hspace{-0.5mm}\Big]\hspace{-1mm}\Bigg),
\end{align}
and if ${P_{po}^{\mathcal{J}}}=1$, which also means that there is no jammer in the scenario, accordingly, $L_{\infty}^{FJ, min}$ is equal to $L_{\infty}^{WoJ}$ as \eqref{Lwoj}.

{{\it Remark 1}: By comparing (\ref{SwoFJ}) and (\ref{finalSinfFJ}), we can obtain $\frac{S_{\infty}^{FJ}}{S_{\infty}^{WoJ}}=2(1-\frac{P_{po}^{\mathcal{J}}}{2})$. This result expresses that the FJ scenario can achieve more high SNR slope compared to the WoJ when a jammer with low threshold to activate the EH circuitry is exploited. Specifically, when~\jj~is always active, FJ achieves twice as the high SNR slope as WoJ. Furthermore, based on  \eqref{finalSinfFJ} which precisely specifies that the power outage at the external jammer has less impact to the high SNR slope rate compared to the power outage at the relay, therefore we can elicit this fact that the jammer's EH component structure can be relatively simple than the relay's.}

\subsubsection{Gaussian Noise Jamming}
In this scenario, the asymptotic ESSR can be obtained as \eqref{abbb}, but by replacing both the expressions $\mathcal{J}_1$ and $\mathcal{J}_2$ indicated in \eqref{Rwfj} with the expressions respectively, given by  
 \begin{align}\label{J1gnj}
{\tilde{\mathcal{{J}}}_1}&\geq\ln(\rho)+\ln(m_xm_y)-2\Phi\nonumber\\
&-\ln\bigg[\beta+m_x+m_u(\xi m_z+m_w)\nonumber\\
&\hspace{9mm}\Big(\frac{\rho m_x}{\beta}+\frac{\ln (\xi m_x)-\ln(m_y)}{\xi m_x-m_y}\Big)\bigg],
\end{align}
and
\begin{align}\label{J2gnj}
{\tilde{\mathcal{{J}}}_2}&\geq\ln(\rho)+\ln(\xi m_xm_y)-2\Phi\nonumber\\
&-\ln\bigg[\beta+m_y+m_u(\xi m_z+m_w)\nonumber\\
&\hspace{9mm}\Big(\frac{\rho m_y}{\beta}+\frac{\ln (\xi m_x)-\ln(m_y)}{\xi m_x-m_y}\Big)\bigg].
\end{align}
The alternative term for $\bar{R}_{Act,\infty}^{FJ}$ in \eqref{abbb} is given by
\begin{equation}\label{rsecgnjact}
\bar{R}_{Act}^{GNJ,\infty}=\frac{1-\alpha}{2\ln 2}\bigg[\ln \Big(\frac{\beta^2}{m_xm_y}\Big)-{\mathcal{J}_3}\bigg].
\end{equation}
Following the similar approach to the FJ in regards of the asymptotic ESSR, the high SNR slope for GNJ can be expressed as
\begin{equation}\label{gnjsinf}
S_{\infty}^{GNJ}=P_{po}^{\mathcal{J}}S_{\infty}^{WoJ}+(1-P_{po}^{\mathcal{R}})(1-P_{po}^{\mathcal{J}})S_{\infty}^{GNJ, Act},
\end{equation}
in which the term $S_{\infty}^{GNJ, Act}$, can be obtained as
\begin{align}\label{sinfgnjact}
S_{\infty}^{GNJ, Act}&=\lim_{\rho\to\infty}\frac{1-\alpha}{2\ln 2}\Bigg[\frac{2\ln \rho -\ln (\frac{\rho^2m_xm_y}{\beta^2})}{\log_2\rho}\Bigg]\stackrel{(a)}{=}0,
\end{align}
where $(a)$ follows from applying L'Hospital's rule to evaluate the limit in the above expression. Finally, substituting \eqref{SwoFJ} and \eqref{sinfgnjact}  into \eqref{gnjsinf} results in
\begin{equation}\label{sinfgnj}
S_{\infty}^{GNJ}=P_{po}^{\mathcal{J}}(1-P_{po}^{\mathcal{R}})\frac{1-\alpha}{2}.
\end{equation}
At this point, we shift our focus to derive the high SNR power offset for the GNJ. Accordingly,  by using \eqref{sinfgnj} and \eqref{SL}, we express $L_{\infty}^{GNJ}$ as
\begin{align}\label{linfgnj}
\hspace{-5mm}L_{\infty}^{GNJ}\hspace{-1.5mm}=\hspace{-2mm}\lim_{\rho\to\infty}\hspace{-1mm}\Bigg[{\log_{2}\rho}-\hspace{-1mm}\Bigg(\hspace{-1mm}\frac{P_{po}^{\mathcal{J}}\bar{R}_{Act}^{WoJ,\infty}\hspace{-1mm}+\hspace{-1mm}(1\hspace{-1mm}-\hspace{-1mm}P_{po}^{\mathcal{J}})\bar{R}_{Act}^{GNJ,\infty}}{P_{po}^{\mathcal{J}}(1-\frac{P_{po}^{\mathcal{J}}}{2})(\frac{1-\alpha}{2})}\hspace{-1mm}\Bigg)\hspace{-1mm}\Bigg]. 
\end{align}
By substituting \eqref{RWoFJin} and \eqref{rsecgnjact} into \eqref{linfgnj}, and after tedious manipulations, we can obtain that $L_{\infty}^{GNJ}=\infty$, which can be concluded intuitively based on the result in \eqref{sinfgnjact}.

\section{Numerical Results and Discussions}

In this section, we provide some numerical examples to verify the accuracy of the provided expressions. Furthermore, we reveal the impact of different system parameters on the ESSR. Two competitive counterparts, the one-way communication \cite{Kalamkar} and the two-way CR aided approach \cite{Xu2015} are used as benchmarks to highlight the secrecy performance of the proposed FJ. In the simulations, unless otherwise stated, we set the following practical system parameters \cite{Kalamkar}.

\begin{table}[t]
	\centering
	\caption{System parameters}
	\resizebox{\columnwidth}{!}{\begin{tabular}[h]{|| c | c | c | c ||}
			\hline
			Parameter & Value & Unit & Description  \\ [1ex]
			\hline\hline
			${P_{\mathcal{S}_{1}}}$ & 10 & dBW & transmit power by~\src  \\
			\hline
			${P_{\mathcal{S}_{2}}}$ & 10 & dBW & transmit power by~\des \\
			\hline
			$\eta$  & 0.7 & - & energy conversion efficiency factor \\
			\hline
			$\Theta$ & 0 & dBm & minimum EH circuitry threshold \\
			\hline
			$N_{0}$ & -10 & dBm & noise power \\
			\hline
			$d_{\mathcal{S}_{1}\mathcal{R}}$ & d=3 & m & \src~$\leftrightarrow$~\rr\quad distance \\
			\hline
			$d_{\mathcal{S}_{2}\mathcal{R}}$ & d=3 & m & \des~$\leftrightarrow$~\rr\quad  distance \\
			\hline
			$d_{S_1J}$ & d=3 & m & \src~$\leftrightarrow$~\jj\quad  distance  \\
			\hline
			$d_{S_2J}$ & d=3 & m & \des~$\leftrightarrow$~\rr\quad  distance \\
			\hline
			$d_{RJ}$ & d=3 & m & \rr~$\leftrightarrow$~\jj\quad  distance   \\
			\hline
			$\kappa$ & 2.7 & - & path loss exponent  \\
			\hline
			$\mu_{ij}$ & $d_{ij}^{-\kappa}$ & - & mean channel power gain\\[1ex] \hline  	
		\end{tabular}}
		
	\label{Table}	
\end{table}

\subsection{Transmit SNR}

\begin{figure}[t]
	\centering
	\includegraphics[width= \columnwidth]{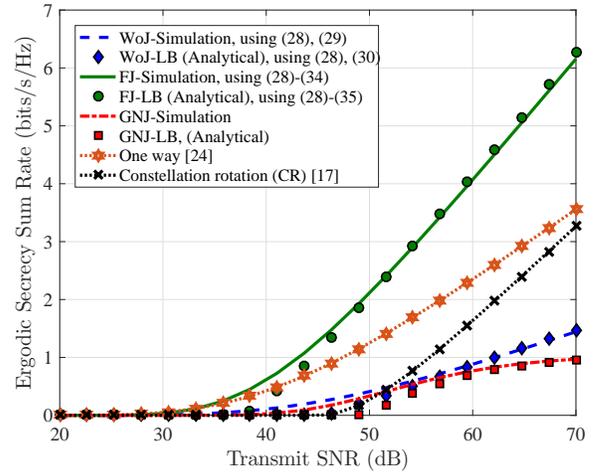}
	\caption{\small ESSR versus transmit SNR for the proposed two-way WoJ, FJ, GNJ, and the one-way communication, as well as the CR approach.}
	\label{fig3}
\end{figure}

Fig. \ref{fig3} plots the ESSR versus transmit SNR for WoJ, FJ, GNJ, and the one-way communication, as well as the CR scheme. From Fig. \ref{fig3}, we observe that the exact numerical expressions for the ESSR of WoJ, given by (\ref{two-wayWoFJ_Exact}), (\ref{Rsecwoj}), and for the ESSR of FJ, given by \eqref{pintfj-total}, \eqref{essr} are well-approximated in the high SNR regime by the closed-form lower-bound expressions in \eqref{two-wayWoFJ_Exact}, (\ref {R_LB}) and (\ref {pintfj-total}), \eqref{R_LB}, respectively.
As can be seen from Fig. \ref{fig3}, only the ESSR of GNJ  is limited a secrecy rate ceiling when the transmit SNR goes beyond a specific threshold, i.e., as predicted before and we observe from Fig. \ref{fig3}. Particularly, the high SNR slop rate for the GNJ scheme is near to zero. That is caused by the fact that although increasing the transmit SNR degrades the received SINR at the relay by augmenting the jamming signal, it also has a detrimental impact on the received SINR at the sources as they can not eliminate the unknown jamming signal. These two contradictory results bring up a saturation region as can be seen from Fig. \ref{fig3}. We can also find from Fig. \ref{fig3} that in the high SNR regime, the proposed two-way FJ substantially outperforms all of its competent counterparts, e.g., in SNR = 50 dB, the ESSR of FJ provides approximately 1 bit/s/Hz more than the one-way transmission scenario even under the assumption of SUD relaying, and is more than twice as much the other two-way benchmarks are. Evidently, from Fig. \ref{fig3}, the high SNR slope of the curve corresponding to the proposed two-way FJ is twice  as much as the slope of the WoJ scenario as we pointed out this result via the mathematical analysis in Remark 1. The last but not least point we need to mention here is that the conventional one-way communication and the CR approaches achieve higher secrecy data rate comparing with WoJ and GNJ in middle-to-high range of SNR, i.e., above SNR = 50 dB, as can be seen from Fig. \ref{fig3}. This observation once again corroborates the idea that how our proposed FJ can dramatically boost the secrecy performance of the system.

\subsection{Time Switching Ratio ($\alpha$)}

\begin{figure}[t]
	\centering
	\includegraphics[width= \columnwidth]{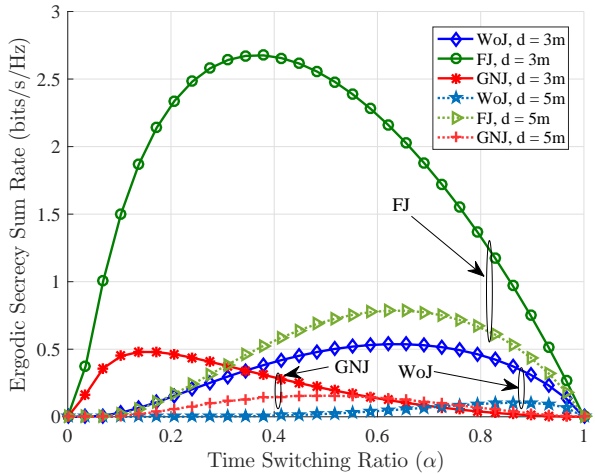}
	\caption{\small ESSR versus TS ratio for WoJ, FJ, and GNJ.  }
	\label{fig4}
\end{figure}

Fig. \ref{fig4}  shows that the ESSR is a quasi-concave function with respect to the TS ratio. For the given system parameters, the maximum ESSR are obtained at the optimum points $\alpha_{opt}^{WoJ} = 0.63$, $\alpha_{opt}^{FJ} = 0.36$, and $\alpha_{opt}^{GNJ} = 0.14$. This finding reveals the importance of TS ratio which should be taken into account in the system design. This observation says that the secrecy performance of the network is highly dependent on both the jamming strategies (WoJ, FJ, or GNJ) and the TS ratio. If the TS ratio is too low, the harvested energy at the relay (and the jammer) may be too low and then, power outage may occur or the received SNR at the sources may be too low. On the other hand, if the TS ratio is too high, insufficient time is dedicated for the relay to broadcast the information signal and hence, the received instantaneous SNR at the receivers may be too low. As a consequence, the reliable communication is influenced. As such, there is a trade-off between a secure transmission and a reliable communication. We consider this issue in our future works. Furthermore, Fig. \ref{fig4} depicts the impact of distance between the network nodes on the ESSR performance. We assume that all the nodes, except the two sources, are located in equal distances from each other denoted by $d$. One interesting result from Fig. \ref{fig4} is that the nodes distance and TS ratio are two proportional parameters subject to the maximum achievable ESSR, i.e., extending the network scale to $d=5$m, the maximum ESSR for all the scenarios is achievable if more time is dedicated to  EH than data relaying. This result is reasonable owning to the fact that by extending the network scale, the path loss phenomenon reduces the received SNR at the relay and the jammer. Therefore, more time should be allocated for EH.

\subsection{Power Allocation Factor ($\lambda$)}

\begin{figure}[t]
	\centering
	\includegraphics[width= \columnwidth]{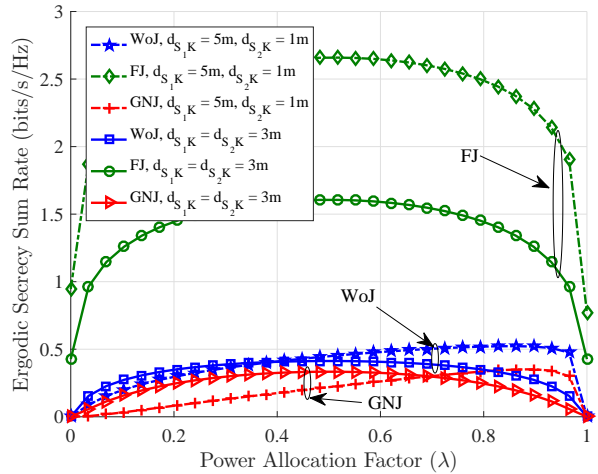}
	\caption{\small ESSR versus power allocation factor for the WoJ, FJ, and GNJ scenarios with respect to the distance of the untrusted relay to the sources. We set $d_{\mathcal{S}_1\mathcal{K}}=3, 5m$, $d_{\mathcal{S}_2\mathcal{K}}=3, 1m$, and $d_{\mathcal{JR}}=1.5m$. Also, $\mathcal{K}$ represents either $\mathcal{R}$ or $\mathcal{J}$. }
	\label{fig5}
\end{figure}
	
We provide Fig. \ref{fig5} to observe the impact of power allocation factor and the relay position with respect to the communication nodes on the achievable ESSR of the two-way WoJ, FJ, GNJ scenarios. Let define the power allocation factor $\lambda$ ($0<\lambda<1$) such that ${P_{\mathcal{S}_{1}}}=\lambda P$ and ${P_{\mathcal{S}_{2}}}=(1-\lambda) P$. We can observe from Fig. \ref{fig5} that for all of the transmission scenarios except the FJ, when the helper nodes are close to either of the communication sources, little amount of the power budget should be allocated to that node to maximize the ESSR. For FJ, regardless of sources distance to the helpers, approximately equal power allocation, i.e., $\lambda\approx0.5$ is required to maximize the ESSR as can be seen from Fig. \ref{fig5}.  It should be pointed out that for the two-way FJ scenario, due to the symmetry of the legitimate nodes' placement, the more closer the source to the untrusted relay should transmit with the less power to provide the higher ESSR. Interestingly, we find that the ESSR performance provided by the GNJ pales in comparison to the WoJ for any power distribution. This observation indicates employing a jammer with unknown jamming signal at the sources, adversely impact  on the communication secrecy.

\subsection{Path Loss Exponent ($\kappa$)}

\begin{figure}[t]
	\centering
	\includegraphics[width= \columnwidth]{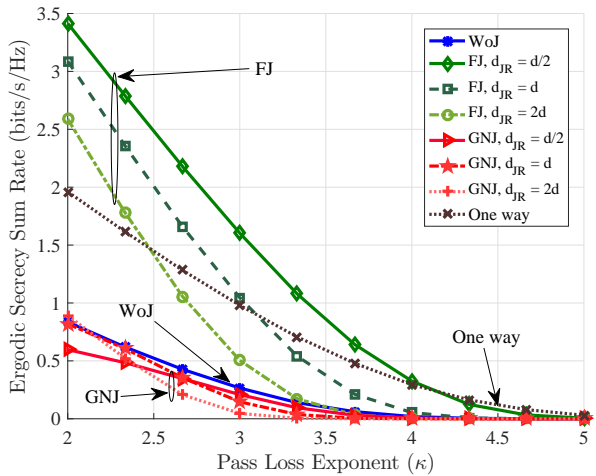}
	\caption{\small ESSR versus the environmental path loss exponent ($\kappa$).}
	\label{fig6}
\end{figure}

We plot Fig. \ref{fig6} to illustrate the impact of path loss exponent on the secrecy rate with different~\jj-to-\rr~distances. When the environmental path loss increases, all the relaying scenarios incontrovertibly suffer from a decline in the ESSR. However, our proposed two-way FJ significantly outperforms the WoJ, GNJ, and one-way communication scenarios. Although, as~\jj's distance to~\rr~increases, the ESSR of  FJ decreases, we can see that the FJ sill presents significantly better ESSR in contrast to the WoJ and GNJ scenarios either in urban (small $\kappa$) or in suburban (large $\kappa$) areas. Furthermore, from another point of view we can draw a conclusion from Fig. \ref{fig6} that by intelligently choosing  the optimal jammer from a group of jammers, e.g., a jammer with low~\jj-to-\rr~distance, the proposed FJ scenario can clearly achieve higher secrecy rate compared to the one-way communication. In addition, we interestingly find that the WoJ scenario outperforms the GNJ. This new result highlights that employing an external jammer with unknown jamming signal brings almost no improvement in terms of the secrecy performance.

\section{Conclusions}

In this paper, we proposed a wireless powered two-way cooperative network in which the two sources communicate via a wireless powered untrusted relay. To enhance the secrecy performance, we proposed to employ an external jammer which is also wirelessly charged by the two sources. By adopting the time switching (TS) protocol at the untrusted relay and jammer, we investigated the ergodic secrecy sum rate (ESSR) of the without jamming (WoJ), friendly jamming (FJ) as well as Guassian noise jamming (GNJ) scenarios. New tight lower-bound expressions were derived for the ESSR of the mentioned secure transmission scenarios, and the asymptotic ESSR analysis to obtain the high SNR slope and the high SNR power offset for the jamming based scenarios, were also presented. Numerical examples revealed the priority of the proposed two-way FJ
compared with the WoJ, GNJ, traditional one-way communication and constellation rotation (CR) aided approaches. Furthermore, several engineering insights were presented regarding the impact of different system parameters such as TS ratio, power allocation factor, path loss exponent, and nodes distance on the ESSR secrecy performances. Our results in this paper gathered new insights to design high rate energy harvesting based networks for device-to-device communications as a part of fifth generation communication network.\\

\section*{appendix A}\label{AppA}

For the wirelessly powered nodes,~\rr~and~\jj, the power outage probability can be written as
\begin{equation}\label{por_app}
P_{po}^{\mathcal{K}}=\Pr\{P_{\mathcal{K}} < \Theta \},
\end{equation}
where substituting \eqref{pr} or \eqref{pj} into \eqref{pop}, one can rewrite \eqref{por_app} as
\begin{equation} \label{por-app-re}
P_{po}^{\mathcal{K}}=\Pr\{{P_{\mathcal{S}_{1}}}|h_{\mathcal{S}_1\mathcal{K}}|^2+{P_{\mathcal{S}_{2}}}|h_{\mathcal{S}_2\mathcal{K}}|^2 < \Theta \}.
\end{equation}
To evaluate $P_{po}^{\mathcal{K}}$, we first present the following useful lemma.
\begin{lemma} \label{Sum}
	Let $S=X+Y$ be a new RV such that $X$ and $Y$ are two exponential RVs with scale parameters $m_x$ and $m_y$, respectively. The PDF and the cumulative distribution function (CDF) of $S$ are as follows
	\begin{equation}\label{fsum}
	f_S(s)=\begin{cases}
	-{\frac { {{\rm e}^{{\frac {s}{m_{x}}}}}-{
				{\rm e}^{{\frac {s}{m_{y}}}}}}{m_{x}-m_{y}}  {{\rm e}^{-{\frac {s \left( m_{x
						}+m_{y} \right) }{m_{x}\,m_{y}}}}}}
	, & m_x\neq m_y \\\\
	{\frac{s}{m^2}}{\rm e}^{-\frac{s}{m}},& m_x=m_y 
	\end{cases}
	\end{equation}
	and
	\begin{align} \label{Fsum}
\hspace{-2mm}F_S(s)\hspace{-1mm}=\hspace{-1mm}\begin{cases}
	1\hspace{-1mm}-\hspace{-1mm}\frac{m_x}{m_x\hspace{-0.5mm}-\hspace{-0.5mm}m_y}{\rm e}^{(-\frac{s}{m_x})} \hspace{-0.5mm}-\hspace{-0.5mm}\frac{m_y}{m_y\hspace{-0.5mm}-\hspace{-0.5mm}m_x}	{\rm e}^{-\frac{s}{m_y}},& m_x \neq m_y \\ \\    \Upsilon(2,\frac{s}{m_x}),& m_x = m_y
	\end{cases}
	\end{align}
	where $\Upsilon(s,x){\hspace {-1mm}}={\hspace {-1mm}}\int_{0}^{x}t^{(s-1)}\emph{e}^{-t} dt$ is the lower incomplete Gamma function \cite{papoulis}. Note that both \eqref{fsum} and \eqref{Fsum} are subjected to the condition $s>0$.
\end{lemma}
\begin{proof}
	We commence from evaluating the PDF of $S$ as 
	\begin{eqnarray}\label{qaz}
	f_S(s){\hspace {-3mm}}&=&{\hspace {-3mm}}\int f_{XY}(x,s-x)dx \nonumber \\
	{\hspace {-3mm}}&\stackrel{(a)}{=}&{\hspace {-3mm}} \int f_{X}(x)f_{Y}(s-x)dx \nonumber \\
	{\hspace {-3mm}}&=&{\hspace {-3mm}}\frac{1}{m_xm_y}{\emph{e}}^{-\frac{s}{m_y}}\int_{0}^{s}\emph{e}^{(\frac{1}{m_y}-\frac{1}{m_x})x}dx,
	\end{eqnarray} 
    where $(a)$ follows from the fact that two RVs $X$ and $Y$ are independent. Finally, evaluating the integral in \eqref{qaz} yields the expression as in \eqref{fsum}, and using the fact that $F_S(s)=\int_{0}^{s}f_S(x) \rm{d}x$, \eqref{Fsum} is also obtained.  
\end{proof}

Using Lemma \ref{Sum} and considering $X{\hspace{-1mm}}=\hspace{-1mm}{P_{\mathcal{S}_{1}}}|h_{\mathcal{S}_1\mathcal{K}}|^2$ and $Y={P_{\mathcal{S}_{2}}}|h_{\mathcal{S}_2\mathcal{K}}|^2$ (which are two exponential RVs with means equal to $m_x$ and $m_y$, respectively) we arrive at $P_{po}^{\mathcal{K}}$ in (\ref{por}) as we know $\Pr\{ X+Y < \Theta_\mathcal{K}\} =F_S(\Theta_\mathcal{K})$.

\section*{appendix B}\label{AppB}

In the following, we proceed to evaluate the terms $I_1$, $I_2$ and $I_3$, respectively. We commence from $I_1$ as follows
\begin{eqnarray}\label{I_1}
I_{1}{\hspace {-3mm}}&=&{\hspace {-3mm}}\E\left\lbrace\ln(1+\gamma_{\mathcal{S}_2})\right\rbrace=\E\left\lbrace\ln(1+\frac{RS}{S+1})\right\rbrace \nonumber\\
&\stackrel{(a)}{\geq}&\ln\left(1+\exp\left(\E\left\lbrace\ln\left( \frac{RS}{S+1}\right) \right\rbrace\label{w}\right)\right) \nonumber \\
{\hspace {-3mm}}&=&{\hspace {-3mm}} \ln\left(1+\exp\left(  \stackrel{\varphi_{1}}{\overbrace{\E \left\lbrace \ln \left[ RS\right] \right\rbrace }} -  \stackrel{\varphi_{2}}{\overbrace{\E \left\lbrace \ln \left[ S+1\right] \right\rbrace }}  \right) \right) \nonumber \\
{\hspace {-3mm}}&\treq&{\hspace {-3mm}}\widehat{I}_{1},
\end{eqnarray}
where $R=\frac{{P_{\mathcal{S}_{1}}}  |h_{\mathcal{S}_1\mathcal{R}}|^2}{N_0}$ and $S=\frac{2\eta\alpha|h_{RS_{2}}|^2}{1-\alpha}$. Furthermore, $(a)$ follows from the fact that $\ln\big(1+\exp(x)\big)$ is a convex function of $x$, since its second derivative is
$\frac{1}{(1+\exp(x))^2}>0$, hence, we can apply Jensen's inequality. It is worth pointing out that the results in \cite{wang2017} express that this lower-bound is sufficiently tight. Using {\cite[Eq. (4.352.1)] {integ}} and {\cite[Eq. (4.331.2)] {integ}}, $\varphi_{1}$ and $\varphi_{2}$ can be calculated, respectively as
\begin{equation}
\varphi_{1}=-2\Phi+\ln\left(m_{R}m_{S}\right),
\end{equation}
and
\begin{equation}
\varphi_{2}=-\exp(\frac{1}{m_{S}})\text{Ei}\left( -\frac{1}{m_{S}}\right).
\end{equation}
Note that the averages of $R$ and $S$ are equal to $m_{R}\hspace{-1mm}=\hspace{-1mm}\frac{{P_{\mathcal{S}_{1}}}\mu_{\mathcal{S}_1\mathcal{R}}}{N_{0}}$ and $m_{S}\hspace{-1mm}=\hspace{-1mm}\frac{2\eta\alpha\mu_{RS_{2}}}{1-\alpha}$, respectively. The term $\widehat{I}_{2}$ is obtained similar to (\ref {I_1}) by replacing $m_{R}=\frac{{P_{\mathcal{S}_{2}}}\mu_{\mathcal{S}_2\mathcal{R}}}{N_{0}}$ and $m_{S}=\frac{2\eta\alpha\mu_{RS_{1}}}{1-\alpha}$. 

Now, attention is shifted to calculate $I_3$ as follows
\begin{eqnarray}
I_3{\hspace{-3mm}}&=&{\hspace{-3mm}}\E\left\lbrace \ln\bigg(1+\gamma_R\bigg)\right\rbrace \nonumber \\
{\hspace{-3mm}}&=&{\hspace{-3mm}} \int_{0}^{\infty}\ln(1+\xi)f_{\gamma_\mathcal{R}}(\xi)~d\xi\nonumber \\
{\hspace{-3mm}}&\stackrel{(a)}{=}&{\hspace{-3mm}}\frac{m_x}{m_y-m_x}\exp\left(\frac{1}{m_x}\right)\mathrm{Ei}\left(-\frac{1}{m_x}\right)\nonumber\\
{\hspace{-3mm}}&+&{\hspace{-3mm}}\frac{m_y}{m_x-m_y}\exp\left(\frac{1}{m_y}\right)\mathrm{Ei}\left(-\frac{1}{m_y}\right),
\end{eqnarray}
where $m_x=\frac{{P_{\mathcal{S}_{1}}}{\mu_{\mathcal{S}_1\mathcal{R}}}}{N_0}$ and $m_y=\frac{{P_{\mathcal{S}_{2}}}{\mu_{\mathcal{S}_2\mathcal{R}}}}{N_0}$, and $(a)$ follows from substituting the PDF of $\gamma_R$  given by \eqref{fsum} and using {\cite[Eq. (4.352.1)] {integ}}.

\section*{appendix C}\label{AppC}
The lower-bound expression for the ESSR of FJ scenario when all the  nodes are active ($\bar{R}_{LB}^{FJ}$) can be obtained as follows
\begin{align}\label{lb}
\bar{R}_{Act}^{FJ}&=\E\left \lbrace  \frac{(1-\alpha)}{2}\left[\log_{2}\frac{(1+\gamma_{\mathcal{S}_2})(1+\gamma_{\mathcal{S}_1})}{(1+\gamma_{\mathcal{R}})}\right]^+\right\rbrace  \nonumber\\ 
&\stackrel{(a)}{\geq}\biggm[\frac{1-\alpha}{2\ln(2)}\bigg( \stackrel{}{\underset{\mathcal{L}_1}{\underbrace{\E\left\lbrace\ln\left(1+\gamma_{\mathcal{S}_2}\right)\right\rbrace}}}+\stackrel{}{\underset{\mathcal{L}_2}{\underbrace{\E\left\lbrace\ln\left(1+\gamma_{\mathcal{S}_1}\right)\right\rbrace}}}  \nonumber\\ &\hspace{5mm}-\stackrel{}{\underset{\mathcal{L}_3}{\underbrace{\E\left\lbrace\ln\left(1+\gamma_{\mathcal{R}}\right)\right\rbrace }}}\bigg)\biggm]^+\treq\bar{R}_{LB}^{FJ}, 
\end{align}
where inequality $(a)$ follows from the fact that $\E\{\max(X,Y)\}\hspace{-3mm}\geq\hspace{-3mm} \max(\E\{X\},\E\{Y\})$ \cite{papoulis}. Moreover, for calculating the part $\mathcal{L}_1$, we first present the following lemma.
\begin{lemma}\label{Lemmazxy}
Let $Z=\frac{M}{N}$ be an arbitrary RV. According to these facts that 1) $\ln(1+x)=\ln(1+\exp(\ln(x)))$, and 2) $\ln(1+\exp(\ln(x)))$ is a convex function with respect to $\ln(x)$, and then applying Jensen's inequality, we can find a tight lower-bound as follows:
\begin{equation}
\E\Big\{\ln(1+Z)\Big\}\hspace{-1mm} \geq \hspace{-1mm} \ln \hspace{-0.5mm}\bigg(1\hspace{-1mm}+\hspace{-0.5mm} \exp\Big[\E\big\{\ln M\big\}\hspace{-1mm}-\hspace{-1mm}\E\big\{\ln N\big\}  \Big]\bigg),
\end{equation}
\end{lemma}
Now, by defining $\gamma_{\mathcal{S}_2}\treq\frac{M}{N}$ in which $M$ and $N$ represent the numerator and denominator of $\gamma_{S_2}$, respectively, and then applying Lemma \ref{Lemmazxy}, we can write
\begin{align} \label{Newformula}
\hspace{-2mm}\E\bigg\{\ln(1+\gamma_{\mathcal{S}_2})\bigg\}
&\hspace{-1mm}\geq\hspace{-1mm}\ln\bigg(1\hspace{-1mm}+\hspace{-0.5mm}\exp\Big[\E\{\ln M\}+\E\{\ln \frac{1}{N}\}\Big]\bigg)\nonumber\\
&\hspace{-2mm}\stackrel{(a)}{\geq}\hspace{-1mm} \ln\bigg(1\hspace{-1mm}+\hspace{-0.5mm}\exp\Big[\stackrel{}{\underset{\mathcal{K}_1}{\underbrace{{\E\{\ln M\}}}}}\Big]\times\frac{1}{\stackrel{}{\underset{\mathcal{K}_2}{\underbrace{{\E\{ N\}}}}}}\bigg), 
\end{align}
where inequality $(a)$ follows from the facts that 1) both the functions $\ln(\cdot)$ and $\exp(\cdot)$ are monotone functions 2) the function $\ln(\frac{1}{x})$ is convex for $x>0$, Therefore, applying Jensen's inequality results in $\E\{\ln\frac{1}{N}\}\geq\ln(\frac{1}{\E\{N\}})$. To obtain $\mathcal{K}_1$, we can further write as
\begin{eqnarray}\label{k1}
\mathcal{K}_1{\hspace {-3mm}}&=&{\hspace {-3mm}}\E\Big\{\ln \big(P_{\mathcal{S}_1}|h_{\mathcal{S}_1\mathcal{R}}|^2 |h_{\mathcal{S}_2\mathcal{R}}|^2\big)\Big\}\nonumber\\
{\hspace {-3mm}}&=&{\hspace {-3mm}}-2\Phi-\ln\bigg(\frac{1}{\bar{\gamma}_{\mathcal{S}_{1}\mathcal{R}}\mu_{\mathcal{S}_2\mathcal{R}}}\bigg),
\end{eqnarray}
where \eqref{k1} follows from Lemma 1. Furthermore, the term $\mathcal{K}_2$ is obtained as
\begin{align}
\mathcal{K}_2&=\E\bigg\{{\hspace{-0.5mm}}{N_{0}\Big(|h_{\mathcal{RS}_{2}}|^2{\hspace{-1mm}}+{\hspace{-1mm}}\frac{\big(P_{\mathcal{S}_1}|h_{\mathcal{S}_{1}\mathcal{J}}|^2\hspace{-1mm}+\hspace{-1mm}P_{S_{2}}|h_{\mathcal{S}_{2}\mathcal{J}}|^2\big)|h_{\mathcal{JR}}|^2}{P_{\mathcal{S}_{1}}|h_{\mathcal{S}_1\mathcal{R}}|^2\hspace{-1mm}+\hspace{-1mm}{P_{\mathcal{S}_{2}}|h_{\mathcal{S}_2\mathcal{R}}|^2}}{\hspace{-1mm}}+{\hspace{-1mm}}\beta\Big)}{\hspace{-1mm}}\bigg\}\nonumber\\
&\stackrel{(a)}{=}{\hspace{-1mm}}N_{0}\bigg[\mu_{\mathcal{R}\mathcal{S}_2}{\hspace{-1mm}}+{\hspace{-1mm}}\beta{\hspace{-1mm}}+{\hspace{-1mm}}\mu_{\mathcal{JR}}\frac{\bar{\gamma}_{\mathcal{S}_{1}\mathcal{J}}\hspace{-1mm}+\hspace{-1mm}\bar{\gamma}_{\mathcal{S}_{2}\mathcal{J}}}{\bar{\gamma}_{\mathcal{S}_{1}\mathcal{R}}\hspace{-1mm}-\hspace{-1mm}\bar{\gamma}_{\mathcal{S}_{2}\mathcal{R}}}\ln\frac{\bar{\gamma}_{\mathcal{S}_{1}\mathcal{R}}}{\bar{\gamma}_{\mathcal{S}_{2}\mathcal{R}}}\bigg],
\end{align}
where $(a)$ follows from the independency of RVs, and using the lemma below.
\begin{lemma}
For two exponential RVs $X$ and $Y$ with the rate parameters $\lambda_x$ and $\lambda_y$, respectively, the new RV $Z=\frac{1}{X+Y}$ with $\lambda_x\neq\lambda_y$ has the following distribution properties	   
\begin{eqnarray}
f_Z(z){\hspace {-3mm}}&=&{\hspace {-3mm}}{\frac {\lambda _{x}\,\lambda _{y}}{{z}^{2} \left( \lambda _{x}-
		\lambda _{y} \right) } \left( -{{\rm e}^{-{\frac {\lambda _{x}}{z}}}}+
	{{\rm e}^{-{\frac {\lambda _{y}}{z}}}} \right) }, \\
F_Z(z){\hspace {-3mm}}&=&{\hspace {-3mm}}{\frac {1}{\lambda _{x}-\lambda _{y}} \left( \lambda _{x}\,{{\rm e}^{{
				\frac {\lambda _{x}}{z}}}}-\lambda _{y}\,{{\rm e}^{{\frac {\lambda _{y
					}}{z}}}} \right) {{\rm e}^{-{\frac {\lambda _{x}+\lambda _{y}}{z}}}}}.
\end{eqnarray}
Moreover, to evaluate $\E\Big\{\frac{1}{X+Y}\Big\}$ one can write as
\begin{eqnarray}
\E\{Z\}{\hspace {-3mm}}&=&{\hspace {-3mm}}\int_{0}^{\infty}zf_Z(z)dz\stackrel{(a)}{=}\int_{0}^{\infty} (1-F_z(z))dz\nonumber\\
{\hspace {-3mm}}&=&{\hspace {-3mm}} {\frac {\lambda _{y}\,\lambda _{x}}{\lambda _{y}-\lambda _{x}}\ln 
	\left( {\frac {\lambda _{y}}{\lambda _{x}}} \right) },
\end{eqnarray} 
where $(a)$ simply follows from integration by part.
\end{lemma}
Finally, the part $\mathcal{L}_1$ is bounded from below by
\begin{equation} \label{vvv}
\mathcal{L}_1\geq\ln\left(1+\exp(\mathcal{K}_1)\big/{\mathcal{K}_2}\right).
\end{equation}
Following the similar steps, $\mathcal{L}_2$ can also be derived, but by simply exchanging the roles of $P_{\mathcal{S}_{1}}$ and $\mu_{\mathcal{S}_1\mathcal{R}}$ with $P_{\mathcal{S}_{2}}$ and $\mu_{\mathcal{S}_2\mathcal{R}}$, respectively.

Finally, we try to find an upper bound for the term $\mathcal{L}_3$ to satisfy the original inequality as well as to find a very tight lower-bound expression for $\bar{R}_{sec}$ which is our primary purpose. To this end, we use the inequality $\E\{\ln(1+x)\}\leq\ln(1+\E\{x\})$ from which $\ln(1+x)$ is a concave function with respect of $x$. By defining 
$X{\hspace{-1mm}}={\hspace{-1mm}}P_{\mathcal{S}_1}|h_{\mathcal{S}_1\mathcal{R}}|^2/N_0$, $Y{\hspace{-1mm}}={\hspace{-1mm}}P_{\mathcal{S}_2}|h_{\mathcal{S}_2\mathcal{R}}|^2/N_0$, $Z{\hspace{-1mm}}={\hspace{-1mm}}\frac{P_{\mathcal{S}_1}|h_{\mathcal{S}_1\mathcal{J}}|^2}{\beta N_0}$, $W{\hspace{-1mm}}={\hspace{-1mm}}\frac{P_{\mathcal{S}_2}|h_{\mathcal{S}_2\mathcal{J}}|^2}{\beta N_0}$, and $U{\hspace{-1mm}}={\hspace{-1mm}}|h_{\mathcal{RJ}}|^2$ as RVs with exponential distribution and means 
$m_x{\hspace{-1mm}}={\hspace{-1mm}}P_{\mathcal{S}_1}\mu_{\mathcal{S}_1\mathcal{R}}/N_0$, 
$m_y{\hspace{-1mm}}={\hspace{-1mm}}P_{\mathcal{S}_2}\mu_{\mathcal{S}_2\mathcal{R}}/N_0$, 
$m_z{\hspace{-1mm}}={\hspace{-1mm}}\frac{P_{\mathcal{S}_1}\mu_{\mathcal{S}_1\mathcal{J}}}{\beta N_0}$,
$m_w{\hspace{-1mm}}={\hspace{-1mm}}\frac{P_{\mathcal{S}_2}\mu_{\mathcal{S}_2\mathcal{J}}}{\beta N_0}$, and $m_u{\hspace{-1mm}}={\hspace{-1mm}}\mu_{\mathcal{R}\mathcal{J}}$, we can express
\begin{eqnarray}
\mathcal{L}_3{\hspace {-3mm}}&=&{\hspace {-3mm}}\E\Big\{\ln(1+\gamma_R)\Big\}\nonumber\\
{\hspace {-3mm}}&=&{\hspace {-3mm}}\E\bigg\{\ln\Big(1+\frac{X+Y}{(Z+W)U+1}\Big)\bigg\}\nonumber\\
{\hspace {-3mm}}&\leq&{\hspace{-3mm}}\ln\bigg(1+\E\Big\{\frac{X+Y}{(Z+W)U+1}\Big\}\bigg)\nonumber\\
{\hspace{-3mm}}&=&{\hspace{-3mm}}\ln\bigg(1+\frac{2(m_x+m_y)}{(m_w-m_z)m_u}\Big[\mathcal{F}_1-\mathcal{F}_2\Big]\bigg),
\end{eqnarray}
where $\mathcal{F}_1$ and $\mathcal{F}_2$ follow from Appendix D.

\section*{appendix D}\label{AppD}

\begin{lemma}\label{NEWLemma}
For two independent RVs $U$ (exponential RV with mean equal to $m_u$) and $S$ (Summation of two independent exponential RVs, i.e., $S=Z+W$ with the PDF and the CDF given in Lemma 2), the new RV $Q=\frac{1}{SU+1}$ has the following distribution properties
\begin{align}
	{\hspace{-3mm}}f_{Q}(q){\hspace{-1mm}}={\hspace{-1.5mm}}\left\{\begin{array}{ll} 
	{\hspace{-1mm}}\frac{2\bigg[K_{0}\Big( \frac{2\sqrt{\frac{1}{q}-1}}{\sqrt{m_{w}m_{u}}}\Big){\hspace{-0.5mm}}-{\hspace{-0.5mm}} K_{0}\Big( \frac{2\sqrt{\frac{1}{q}-1}}{\sqrt{m_{z}m_{u}}}\Big)\bigg]} {q^2(m_{w}-m_{z})m_{u}},& 0<q\leq1  \\
	0,&\text{o.w.}  
	\end{array}
	\right.
\end{align}
and then, for $q \in (0, 1]$ we have
\begin{align}
{\hspace{-3mm}}F_{Q}(q)&{\hspace{-0.5mm}}={\hspace{-0.5mm}}
1{\hspace{-1mm}}+{\hspace{-1mm}}\frac{2m_z}{m_u(m_w{\hspace{-0.5mm}}-{\hspace{-0.5mm}}m_z)}\sqrt{\frac{m_u(1-q)}{m_zq}} K_{1}\Big( \frac{2\sqrt{1-{q}}}{\sqrt{m_{z}m_{u}q}}\Big)\nonumber\\
{\hspace{-3mm}}&-{\hspace{-0.5mm}}\frac{2m_w}{m_u(m_w{\hspace{-0.5mm}}-{\hspace{-0.5mm}}m_z)}\sqrt{\frac{m_u(1-q)}{m_wq}} K_{1}\Big( \frac{2\sqrt{1-{q}}}{\sqrt{m_{w}m_{u}q}}\Big).
\end{align}
\end{lemma}
\begin{proof}
Let commence from the definition of CDF
\begin{eqnarray}\label{AXWQ}
F_Q(q){\hspace {-3mm}}&=&{\hspace {-3mm}}\Pr\Big\{Q\leq q\Big\}\nonumber\\
{\hspace {-3mm}}&=&{\hspace {-3mm}}\E_{u}\Big\{\Pr\big\{S\leq\frac{1-q}{qu}\Big|U=u\big\}\Big\}\nonumber\\
{\hspace {-3mm}}&=&{\hspace {-3mm}}\E_{u}\Big\{F_S\big(\frac{1-q}{qu}\big)\Big\}\nonumber\\
{\hspace {-3mm}}&\stackrel{(a)}{=}&{\hspace {-3mm}}\int_{0}^{\infty}
\bigg(1+\frac{m_z}{m_w-m_z}\exp(-\frac{1-q}{m_zqu}) \nonumber\\
{\hspace {-3mm}}&-&{\hspace {-3mm}}\frac{m_w}{m_w-m_z}\exp(-\frac{1-q}{m_wqu})\bigg)
 \frac{\exp(\frac{-u}{m_u})}{m_u} du,
\end{eqnarray}
where $(a)$ follows from Appendix A. Also, the last equality can be further calculated using {\cite[Eq. (3.471.9)] {integ}}  with
\begin{equation}
\int_{0}^{\infty} x^{\nu-1}\exp(-\alpha x- \frac{\beta}{x})\mathrm{d}x=2 \left(\frac{\beta}{\alpha}\right)^{\frac{\nu}{2}}K_{\nu}\left(2\sqrt{\alpha \beta}\right),
\end{equation}
where $K_{\nu}(\cdot)$ is the modified Bessel function of the second kind and $\nu$-th order. Finally, after simple manipulations as well as using the fact that $\E\{X\}=\int_{0}^{\infty}(1-F_{X}(x))dx$, we can evaluate $\E\{Q\}$ as
\begin{align}
{\hspace{-3mm}}\E\{Q\}&{\hspace{-0.5mm}}={\hspace{-0.5mm}}\int_{0}^{\infty}qf_{Q}(q)\mathrm{d}q\nonumber\\
&{\hspace{-0.5mm}}={\hspace{-0.5mm}}\int_{0}^{1}\frac{2\bigg[K_{0}\Big( \frac{2\sqrt{\frac{1}{q}-1}}{\sqrt{m_{w}m_{u}}}\Big) - K_{0}\Big( \frac{2\sqrt{\frac{1}{q}-1}}{\sqrt{m_{z}m_{u}}}\Big)\bigg]} {q(m_{w}-m_{z})m_{u}}\mathrm{d}q\nonumber\\
&{\hspace{-0.5mm}}={\hspace{-0.5mm}}\frac{2}{(m_w-m_z)m_u}\nonumber\\
&{\hspace{-0.5mm}}\times{\hspace{-0.5mm}}\Bigg[
{\stackrel{}{\underset{\mathcal{F}_1}{\underbrace{\int_{0}^{1}\frac{K_{0}\Big( \frac{2\sqrt{\frac{1}{q}-1}}{\sqrt{m_{w}m_{u}}}\Big) }{q}\mathrm{d}q}}}}{\hspace{-0.5mm}}-{\hspace{-0.5mm}}
{\stackrel{}{\underset{\mathcal{F}_2}{\underbrace{\int_{0}^{1}\frac{K_{0}\Big( \frac{2\sqrt{\frac{1}{q}-1}}{\sqrt{m_{z}m_{u}}}\Big) }{q}\mathrm{d}q}}}}\Bigg].
\end{align}
Now, due to symmetry of the integrals in the last equation, we only compute the first term $\mathcal{F}_1$ as
\begin{equation}\label{FFF}
\mathcal{F}_1=\int_{0}^{1}\frac{K_{0}\Big(C_1\sqrt{\frac{1}{q}-1}\Big) }{q}\mathrm{d}q,
\end{equation}
where $C_1$ is defined as $C_1=\frac{2}{\sqrt{m_wm_u}}$. To proceed further, we employ an equivalent definition, i.e., an infinite series of modified Bessel functions of the second kind and $\nu$-th order, with $\nu>0$, as introduced in \cite{Molu2017}
\begin{equation}\label{kn}
\mathrm{K}_{\nu}(\beta x)=\exp(-\beta x)\sum\limits_{n=0}^{\infty}\sum_{i=0}^{n}\Lambda(\nu, n, i)(\beta x)^{i-\nu},
\end{equation}
where
\begin{equation}
\Lambda(\nu, n, i)= \frac{(-1)^i\sqrt{\pi}\Gamma(2\nu)\Gamma(n-\nu+\frac{1}{2})\mathrm{L}(n, i)}{2^{\nu-i}\Gamma(\frac{1}{2}-\nu)\Gamma(n+\nu+\frac{1}{2})n!}.
\end{equation}
However, we cannot directly apply the expression in \eqref{kn} to calculate the integral in \eqref{FFF} since in our case, we have $\nu=0$. Therefore, using the equality 
$\mathrm{K}_{\nu-2}(\beta x)=\mathrm{K}_{\nu}(\beta x)-\frac{2(\nu-1)}{\beta x}\mathrm{K}_{\nu-1}(\beta x)$ \cite{Molu2017}
  when $\nu=2$, we can rewrite $\mathrm{K}_{0}(\beta x)$ as
\begin{equation}\label{K0app}
\mathrm{K}_{0}(\beta x)\hspace{-1mm}=\hspace{-1mm}\exp(-\beta x)\sum\limits_{n=0}^{\infty}\sum_{i=0}^{n}\Lambda(1, n, i)(g(n)\hspace{-1mm}-\hspace{-1mm}2)(\beta x)^{i-2},
\end{equation}
where
\footnote{Note that for the evaluation of the coefficients, $\Lambda(\nu, n, i)$, following results are fruitful: $L(0, 0)=1$, $L(n, 0)=0$, $L(n, 1)=n!$ for positive values of $n$. In addition, for Gamma function $\Gamma(\frac{1}{2})=\sqrt{\pi}$, $\Gamma(-\frac{1}{2})=-2\sqrt{\pi}$, $\Gamma(1)=1$, and $\Gamma(x+1)=x\Gamma(x)$.}
\begin{equation}
g(n)=\frac{\Lambda(2, n, i)}{\Lambda(1, n, i)}=-\frac{9}{2}\frac{\Gamma(n-\frac{3}{4})\Gamma(n+\frac{3}{2})}{\Gamma(n-\frac{1}{2})\Gamma(n+\frac{5}{2})}.
\end{equation}
Substituting \eqref{K0app} in \eqref{FFF}, and after some manipulations, one can represent \eqref{FFF} as
\begin{eqnarray}
\mathcal{F}_1{\hspace {-3mm}}&=&{\hspace {-3mm}}\sum\limits_{n=1}^{\infty}\sum_{i=1}^{n}\Lambda(1, n, i)(g(n)-2)C_1^{i-2}\nonumber\\
{\hspace {-3mm}}&\times&{\hspace {-3mm}}\int_{0}^{1}\frac{\exp(-C_1(\frac{1-q}{q}))(\sqrt{\frac{1-q}{q}})^{i-2}}{q}\mathrm{d}q\nonumber\\
{\hspace {-3mm}}&\stackrel{(a)}{=}&{\hspace {-3mm}}2\int_{0}^{\infty}\frac{\exp(-C_1u)u^{i-1}}{u^2+1}\mathrm{d}u,
\end{eqnarray}
where $(a)$ follows from taking the axillary variable $u\hspace{-1.25mm}=\hspace{-1.25mm}\sqrt{\frac{1-q}{q}}$.
Using {\cite[Eq. (3.356.1)] {integ}} and {\cite[Eq. (3.356.2)] {integ}} we ultimately obtain the required expression for $\mathcal{F}_{1,2}$ as
\begin{align}
\mathcal{F}_{1, 2}(x)&{\hspace {-1mm}}={\hspace {-1mm}}-\sum_{n=1}^{\infty}\sum_{i=1}^{n}\Lambda(1,n,i)\bigg(9
\frac{\Gamma(n{\hspace {-1mm}}-{\hspace {-1mm}}\frac{3}{4})\Gamma(n{\hspace {-1mm}}+{\hspace {-1mm}}\frac{3}{2})}{\Gamma(n{\hspace {-1mm}}-{\hspace {-1mm}}\frac{1}{2})\Gamma(n{\hspace {-1mm}}+{\hspace {-1mm}}\frac{5}{2})}{\hspace {-0.5mm}}+{\hspace {-0.5mm}}4\bigg)\nonumber\\
&\hspace{-10mm}\times x^{i-2} \begin{cases}
(-1)^k\Big[\mathrm{ci}(x)\cos(x){\hspace {-0.5mm}}+{\hspace {-0.5mm}}\mathrm{si}(x)\sin(x)\Big]\nonumber\\
+{\hspace {-0.5mm}}\frac{1}{x^{2k-2}}\sum_{j=1}^{k-1}(2k{\hspace {-0.5mm}}-{\hspace {-0.5mm}}2j{\hspace {-0.5mm}}-{\hspace {-0.5mm}}1)!(-x^2)^{j-1},& i=2k\\\\
(-1)^k\Big[\mathrm{ci}(x)\sin(x){\hspace {-0.5mm}}-{\hspace {-0.5mm}}\mathrm{si}(x)\cos(x)\Big]\nonumber\\
+\frac{1}{x^{2k-1}}\sum_{j=1}^{k}(2k{\hspace {-0.5mm}}-{\hspace {-0.5mm}}2j)!(-x^2)^{2j-1},& i=2k+1
\end{cases}\\
\end{align}
It should be clearly noted that $\mathcal{F}_{1, 2}(C_1)=\mathcal{F}_1$, and $\mathcal{F}_{1, 2}(C_2)=\mathcal{F}_2$, which $C_2=\frac{2}{\sqrt{m_zm_u}}$.
\end{proof}


%
%

\section*{Acknowledgment}
The authors would like to thank Dr. Phee Lep Yeoh for the constructive comments to improve the paper and indispensable collaboration that led to our joint prior work.

\end{document}